\documentclass[times, 11pt]{article}
\usepackage[thmnum]{preamble}
\usepackage{subcaption}

\newcommand{\free}{\ensuremath{\mathrm{free}}}
\newcommand{\fix}{\ensuremath{\mathrm{fix}}}
\renewcommand{\boolcube}{\ensuremath{\{0,1\}}}
\newcommand{\depth}{\ensuremath{\mathrm{depth}}}
\newcommand{\Res}{\ensuremath{\mathsf{Res}}}

\renewcommand{\restriction}{\!\ensuremath{\upharpoonright}\!}
\newcommand{\sCP}{\ensuremath{\mathsf{sCP}}}

\begin{document}
\newgeometry{margin=1.4in,top=1.6in,bottom=1in}

\begin{center}
{\LARGE On the Power and Limitations of Branch and Cut}
\\[1cm] \large

\setlength\tabcolsep{0em}
\newcommand{\myPad}{\hspace{2.9em}}
\begin{tabular}{c@{\myPad}c@{\myPad}c}
  Noah Fleming &
	Mika G\"o\"os$^\dagger$  &
	Russell Impagliazzo \\[-.5mm]
  \small\slshape  University of Toronto &
	\small\slshape EPFL &
	\small\slshape University of California,\\[-1mm]
	\small\slshape \& Simons Institute & & \small\slshape San Diego\\[3mm]
  Toniann Pitassi &
	Robert Robere$^\dagger$ &
	Li-Yang Tan \\[-.5mm]
  \small\slshape University of Toronto \& IAS &
	\small\slshape McGill University &
  \small\slshape Stanford University \\[-1mm]
  \\[3mm]
  & Avi Wigderson & \\[-.5mm]
  & \small\slshape IAS &
\end{tabular}

\vspace{9mm}

\large
{\today}

\vspace{9mm}

\bf Abstract

\end{center}

\normalsize
\noindent

The Stabbing Planes proof system \cite{BeameFIKPPR18} was introduced to model the reasoning carried out in practical mixed integer programming solvers. As a proof system, it is powerful enough to simulate Cutting Planes and to refute the Tseitin formulas 
--- certain unsatisfiable systems of linear equations $\bmod 2$ --- which are canonical hard examples for many algebraic proof systems. In a recent (and surprising) result, Dadush and Tiwari \cite{TseitinUpperBound} showed that these short refutations of the Tseitin formulas could be translated into quasi-polynomial size and depth Cutting Planes proofs, refuting a long-standing conjecture. 
This translation raises several interesting questions. First, whether all Stabbing Planes proofs can be efficiently simulated by Cutting Planes. This would allow for the substantial analysis done on the Cutting Planes system to be lifted to practical mixed integer programming solvers. Second, whether the quasi-polynomial depth of these proofs is inherent to Cutting Planes. 

In this paper we make progress towards answering both of these questions. First, we show that \emph{any} Stabbing Planes proof with bounded coefficients ($\SP^*$) can be translated into Cutting Planes. As a consequence of the known lower bounds for Cutting Planes, this establishes the first exponential lower bounds on $\SP^*$. Using this translation, we extend the result of Dadush and Tiwari to show that Cutting Planes has short refutations of any unsatisfiable system of linear equations over a finite field. Like the Cutting Planes proofs of Dadush and Tiwari, our refutations also incur a quasi-polynomial blow-up in depth, and we conjecture that this is inherent. As a step towards this conjecture, we develop a new \emph{geometric} technique for proving lower bounds on the depth of Cutting Planes proofs. This allows us to establish the first lower bounds on the depth of \emph{Semantic} Cutting Planes proofs of the Tseitin formulas.

\renewcommand*{\thefootnote}{\fnsymbol{footnote}}
\footnotetext[2]{Work done while at Institute for Advanced Study.}
\thispagestyle{empty}
\setcounter{page}{0}
\newpage
\restoregeometry

\section{Introduction}
\label{sec:intro}

An effective method for analyzing classes of algorithms is to formalize the techniques used by the class into a \emph{formal proof system}, and then analyze the formal proof system instead. 
By doing this, theorists are able to hide many of the practical details of implementing these algorithms, while preserving the class of methods that the algorithms can feasibly employ.
Indeed, this approach has been applied to study many different families of algorithms, such as 
\begin{itemize}
	\item \emph{Conflict-driven clause-learning} algorithms for SAT \cite{BayardoS97,MarquesS99,MoskewiczMZZM01}, which can be formalized using \emph{resolution} proofs \cite{DavisP60}.
	\item Optimization algorithms using \emph{semidefinite programming} \cite{Goemans1994,parrilo2000}, which can often be formalized using \emph{Sums-of-Squares proofs} \cite{dima-sos,BarakBHKSZ12}.
	\item The classic \emph{cutting planes} algorithms for integer programming \cite{gomory1963algorithm,Chvatal73a}, which are formalized by \emph{cutting planes proofs} \cite{Chvatal73a, Chvatal84, CookCT87}.
\end{itemize}

In the present work, we continue the study of formal proof systems corresponding to modern integer programming algorithms.
Recall that in the integer programming problem, we are given a polytope $P \subseteq \reals^n$ and a vector $c \in \reals^n$, and our goal is to find a point $x \in P \cap \integers^n$ maximizing $c \cdot x$.
The classic approach to solving this problem --- pioneered by Gomory \cite{gomory1963algorithm} --- is to add\footnote{Throughout, we will say that a cutting plane, or an inequality is \emph{added} to a polytope $P$ to mean that it is added to the set of inequalities defining $P$.} \emph{cutting planes} to  $P$.
A \emph{cutting plane} for $P$ is any inequality of the form $ax \leq \lfloor b \rfloor$, where $a$ is an integral vector, $b$ is rational, and \emph{every} point of $P$ is satisfied by $ax \leq b$.
By the integrality of $a$, it follows that cutting planes \emph{preserve} the integral points of $P$, while potentially \emph{removing} non-integral points from $P$.
The cutting planes algorithms then proceed by heuristically choosing ``good'' cutting planes to add to $P$ to try and locate the integral hull of $P$ as quickly as possible.

As mentioned above, these algorithms can be naturally formalized into a proof system --- the \emph{Cutting Planes proof system}, denoted $\CP$ --- as follows \cite{CookCT87}.
Initially, we are given a polytope $P$, presented as a list of integer-linear inequalities $\set{a_i x \leq b_i}$.
From these inequalities we can then deduce new inequalities using two deduction rules:
\begin{itemize}
	\item \emph{Linear Combination.} From inequalities $ax \leq b, cx \leq d$, deduce any non-negative linear combination of these two inequalities with integer coefficients.
	\item \emph{Division Rule.} From an inequality $ax \leq b$, if $d \in \integers$ with $d \geq 0$ divides all entries of $a$ then deduce $(a/d)x \leq \lfloor b/d \rfloor$.
\end{itemize}
A Cutting Planes \emph{refutation} of $P$ is a proof of the trivially false inequality $1 \leq 0$ from the inequalities in $P$; clearly, such a refutation is possible only if $P$ does not contain any integral points.
While Cutting Planes has grown to be an influential proof system in propositional proof complexity, the original cutting planes algorithms suffered from numerical instabilities, as well as difficulties in finding good heuristics for the next cutting planes to add \cite{gomory1963algorithm}.

The modern algorithms in integer programming improve on the classical cutting planes method by combining them with a second technique, known as \emph{branch-and-bound}, resulting in a family of optimization algorithms broadly referred to as \emph{branch-and-cut algorithms}. 
These algorithms search for integer solutions in a polytope $P$ by recursively repeating the following two procedures: First, $P$ is split into smaller polytopes $P_1,\ldots, P_k$  such that $P \cap \mathbb{Z}^n \subseteq \bigcup_{i\in [k]} P_i$ (i.e.~\emph{branching}). 
Next, cutting planes deductions are made in order to further refine the branched polytopes (i.e.~\emph{cutting}). 
In practice, branching is usually performed by selecting a variable $x_i$ and branching on all possible values of $x_i$; that is, recursing on $P \cap \{x_i = t\}$ for each feasible integer value $t$.
More complicated branching schemes have also been considered, such as branching on the hamming weight of subsets of variables \cite{FischettiL03}, branching using basis-reduction techniques~\cite{AardalL04,KrishnamoorthyP09,AardalBHLS00}, and more general linear inequalities \cite{OwenM01,mahajanRT09,KaramanovC11}. 

However, while these branch-and-cut algorithms are much more efficient in practice than the classical cutting planes methods, they are no longer naturally modelled by Cutting Planes proofs.
So, in order to model these solvers as proof systems, Beame et al.~\cite{BeameFIKPPR18} introduced the \emph{Stabbing Planes} proof system.
Given a polytope $P$ containing no integral points, a \emph{Stabbing Planes} refutation of $P$ proceeds as follows.
We begin by choosing an integral vector $a$, an integer $b$, and replacing $P$ with the two polytopes $P \cap \set{ax \leq b-1}$ and $P \cap \set{ax \geq b}$.
Then, we recurse on these two polytopes, continuing until all descendant polytopes are empty (that is, they do not even contain any \emph{real} solutions).
The majority of branching schemes used in practical branch-and-cut algorithms (including all of the concrete schemes mentioned above) are examples of this general branching rule.

It is now an interesting question how the two proof systems --- Cutting Planes and Stabbing Planes --- are related.
By contrasting the two systems we see at least three major differences:
\begin{itemize}
	\item \emph{Top-down vs.~Bottom-up.} Stabbing Planes is a \emph{top-down} proof system, formed by performing queries on the polytope and recursing; while Cutting Planes is a \emph{bottom-up} proof system, formed by deducing new inequalities from old ones.
	\item \emph{Polytopes vs.~Halfspaces.} Individual ``lines'' in a Stabbing Planes proof are \emph{polytopes}, while individual ``lines'' in a Cutting Planes proof are \emph{halfspaces}.
	\item \emph{Tree-like vs.~DAG-like.} The graphs underlying Stabbing Planes proofs are trees, while the graphs underlying Cutting Planes proofs are general DAGs: intuitively, this means that Cutting Planes proofs can ``re-use'' their intermediate steps, while Stabbing Planes proofs cannot.
\end{itemize}
When taken together, these facts suggest that Stabbing Planes and Cutting Planes could be incomparable in power, as polytopes are more expressive than halfspaces, while DAG-like proofs offer the power of line-reuse.
Going against this natural intuition, Beame et al.~proved that Stabbing Planes \emph{can} actually efficiently simulate Cutting Planes \cite{BeameFIKPPR18} (see \autoref{fig:relationships}) --- this simulation was later extended by Basu et al.~\cite{BasuCDJ21} to almost all types of cuts used in practical integer programming, including split cuts.
Furthermore, Beame et al. proved that Stabbing Planes is \emph{equivalent} to the proof system \emph{tree-like $\mathsf{R}(\CP)$}, denoted $\tRCP$, which was introduced by Kraj\'{i}\v{c}ek \cite{Krajicek98}, and whose relationship to Cutting Planes was previously unknown.

This leaves the converse problem --- of whether Stabbing Planes can also be simulated by Cutting Planes --- as an intriguing open question.
Beame et al.~conjectured that such a simulation was impossible, and furthermore that the \emph{Tseitin formulas} provided a separation between these systems \cite{BeameFIKPPR18}.
For any graph $G$ and any $\set{0,1}$-labelling $\ell$ of the vertices of $G$, the \emph{Tseitin formula} of $(G, \ell)$ is the following system of $\mathbb{F}_2$-linear equations: for each edge $e$ we introduce a variable $x_e$, and for each vertex $v$ we have an equation
\[ \bigoplus_{u: uv \in E} x_{uv} = \ell(v)\]
asserting that the sum of the edge variables incident with $v$ must agree with its label $\ell(v)$ (note such a system is unsatisfiable as long as $\sum_v \ell(v)$ is odd).
On the one hand, Beame et al.~proved that there are \emph{quasi-polynomial size} Stabbing Planes refutations of the Tseitin formulas \cite{BeameFIKPPR18}.
On the other hand, Tseitin formulas had long been conjectured to be exponentially hard for Cutting Planes~\cite{CookCT87}, as they form one of the canonical families of hard examples for algebraic and semi-algebraic proof systems, including Nullstellensatz \cite{dima-nsatz}, Polynomial Calculus \cite{bgip}, and Sum-of-Squares \cite{dima-sos,Schoenebeck08}.

In a recent breakthrough, the long-standing conjecture that Tseitin was exponentially hard for Cutting Planes was \emph{refuted} by Dadush and Tiwari~\cite{TseitinUpperBound}, who gave \emph{quasi-polynomial size} Cutting Planes refutations of Tseitin instances.  
Moreover, to prove their result, Dadush and Tiwari showed how to \emph{translate} the quasipolynomial-size Stabbing Planes refutations of Tseitin into Cutting Planes refutations.
This translation result is interesting for several reasons. 
First, it brings up the possibility that Cutting Planes \emph{can actually} simulate Stabbing Planes.
If possible, such a simulation would allow the significant analysis done on the Cutting Planes system to be lifted directly to branch-and-cut solvers. 
In particular, this would mean that the known exponential-size lower bounds for Cutting Planes refutations would immediately imply the first exponential lower bounds for these algorithms for arbitrary branching heuristics.
Second, the translation converts {\it shallow} Stabbing Planes proofs into \emph{very deep} Cutting Planes proofs: the Stabbing Planes refutation of Tseitin has depth $O(\log^2 n)$ and quasi-polynomial size, while the Cutting Planes refutation has quasipolynomial size \emph{and} depth.
This is quite unusual since simulations between proof systems typically preserve the structure of the proofs, and thus brings up the possibility that the Tseitin formulas yield a {\it supercritical} size/depth tradeoff -- formulas with short proofs, requiring {\it superlinear} depth.
For contrast: another simulation from the literature which emphatically does \emph{not} preserve the structure of proofs is the simulation of \emph{bounded-size} resolution by \emph{bounded-width} resolution by Ben-Sasson and Wigderson \cite{BenSassonW01}.
In this setting, it is known that this simulation is tight \cite{BonetG01}, and even that there exist formulas refutable in resolution width $w$ requiring maximal size $n^{\Omega(w)}$ \cite{AtseriasLN16}.
Furthermore, under the additional assumption that the proofs are \emph{tree-like}, Razborov \cite{Razborov16} proved a supercritical trade-off between width and size.



\subsection{Our Results}

\subsubsection*{A New Characterization of Cutting Planes}

Our first main result gives a \emph{characterization} of Cutting Planes proofs as a natural subsystem of Stabbling Planes that we call \emph{Facelike} Stabbing Planes.
A Stabbing Planes query is \emph{facelike} if one of the sets $P \cap \{ax \leq b-1\}$ or $P \cap \{ax \geq b\}$ is either empty or is a face of the polytope $P$, and a Stabbing Planes proof is said to be facelike if it only uses facelike queries. 
Our main result is the following theorem. 

\begin{restatable}{thm}{main}
	\label{thm:main}
	The proof systems $\CP$ and Facelike $\SP$ are polynomially equivalent.
\end{restatable}

The proof of this theorem is inspired by Dadush and Tiwari's upper bound for the Tseitin formulas. Indeed, the key tool underlying both their proof and ours is a lemma due to Schrijver~\cite{SCHRIJVER1980291} which allows us to simulate $\CP$ refutations of faces of a polytope, when beginning from $P$ itself.

Using this equivalence we prove the following surprising simulation (see \autoref{fig:relationships}), stating that Stabbing Planes proofs with relatively small coefficients (quasi-polynomially bounded in magnitude) can be quasi-polynomially simulated by Cutting Planes.

\begin{restatable}{thm}{quasi}
	\label{thm:quasipoly-simulation}
	Let $F$ be any unsatisfiable CNF formula on $n$ variables, and suppose that there is a $\SP$ refutation of $F$ in size $s$ and maximum coefficient size $c$. 
	Then there is a $\CP$ refutation of $F$ in size $s(cn)^{\log s}$.
\end{restatable}



In fact, we prove a more general result (\autoref{thm:low-coefficients}) which holds for arbitrary polytopes $P \in \mathbb{R}^n$, rather than only for CNF formulas, which degrades with the \emph{diameter} of $P$. This should be contrasted with the work of Dadush and Tiwari~\cite{TseitinUpperBound}, who show that any $\SP$ proof of size $s$ of a polytope with diameter $d$ can be assumed to have coefficients of size $(nd)^{O(n^2)}$.

As a second application of \autoref{thm:main}, we generalize Dadush and Tiwari upper bound for Tseitin to show that Cutting Planes can refute any unsatisfiable system of linear equations over a finite field. 
This follows by showing that, like Tseitin, we can refute such systems of linear equations in quasi-polynomial-size Facelike $\SP$.
\begin{thm}\label{thm:linear-equations}
	Let $F$ be the CNF encoding of an unsatisfiable system of $m$ linear equations over a finite field. 
	There is a $\CP$ refutation of $F$ of size $|F|^{O(\log m)}$. 
\end{thm}

This should be contrasted with the work of Filmus, Hrube\v{s}, and Lauria \cite{FilmusHL16}, which gives several unsatisfiable systems of linear equations over $\reals$ that require \emph{exponential size} refutations in Cutting Planes (see \autoref{fig:relationships}).

\begin{figure}
	\centering 
		\hspace{3em}\begin{tikzpicture}
		
 		\draw[color=black!80, very thick, rounded corners=0.5ex,fill=green!9] (-3.7,8.3) -- (-3.7,7.7) -- (-0.6,7.7) -- (-0.6,8.3) -- cycle ;
 		\node[text width=3cm] at (-2,8) {$\SP = \tRCP$};
 		
 		\draw[color=black!80, very thick, rounded corners=0.5ex,fill=green!9] (0.7,8.3) -- (0.7,7.7) -- (3.1,7.7) -- (3.1,8.3) -- cycle ;
 		\node[text width=3cm] at (2.4,8) {Semantic $\CP$};
 
 		\draw[color=black!80, very thick, rounded corners=0.5ex,fill=green!9] (0.3,6.8) -- (0.3,6.2) -- (3.5,6.2) -- (3.5,6.8) -- cycle ;
 		\node[text width=3cm] at (2,6.5) {$\CP = $ Facelike $\SP$};

 		\draw[color=black!80, very thick, rounded corners=0.5ex,fill=green!9] (-2.4,6.8) -- (-2.4,6.2) -- (-1.1,6.2) -- (-1.1,6.8) -- cycle ;
 		\node[text width=1cm] at (-1.5,6.5) (text1){$\SP^*$};
 		
 		\draw[color=black!80, very thick, rounded corners=0.5ex,fill=green!9] (-0.9,5.3) -- (-0.9,4.7) -- (0.4,4.7) -- (0.4,5.3) -- cycle ;
 		\node[text width=1cm] at (0,5) (text1){$\CP^*$};

 		
  		\draw[->, color=black!80, very thick] (-1.7,6.9) -- (-1.7,7.6);
  		\draw[->, color=black!80, very thick] (2.3,6.9) -- (2.3,7.6);
  		\draw[<-, color=red!45, dashed, very thick] (1.3,6.9) -- (1.3,7.6);
  		\draw[<-, color=red!45, very thick] (0.2,6.5) -- (-1,6.5);
  		\draw[->, color=black!80, very thick] (0.2,5.4) -- (1.5,6.1);
  		\draw[->, color=black!80, very thick] (-0.7,5.4) -- (-1.5,6.1);
  		
  		\draw[->, color=black!80, very thick] (0.3,6.9) -- (-0.6,7.6);
  		
  		\draw[<-, color=red!45, dashed, very thick] (-0.5,8) -- (0.6,8);
 	\end{tikzpicture}
	\caption{Known relationships between proof systems considered in this paper. A solid black (red) arrow from proof system $P_1$ to $P_2$ indicates that $P_2$ can polynomially (quasi-polynomially) simulate $P_1$. A dashed arrow indicates that this simulation cannot be done.} \label{fig:relationships}
\end{figure}
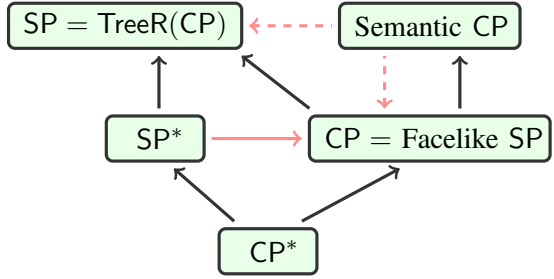


\subsubsection*{Lower Bounds} 
An important open problem is to prove superpolynomial size lower bounds for Stabbing Planes proofs.
We make significant progress toward this goal by proving the first superpolynomial lower bounds
on the size of low-weight Stabbing Planes proofs.
Let $\SP^*$ denote the family of Stabbing Planes proofs in which each coefficient has at most quasipolynomial ($n^{\log^{O(1)} n}$) magnitude.

\begin{thm}\label{thm:spstar-lb}
   There exists a family of unsatisfiable CNF formulas $\{F_n\}$ such that any $\SP^*$ refutation of $F$ requires size at least $2^{n^\varepsilon}$ for constant $\varepsilon > 0$.
\end{thm}

Our proof follows in a straightforward manner from \autoref{thm:quasipoly-simulation} together with known Cutting Planes lower bounds.
We view this as a step toward proving SP lower bounds (with no restrictions on the weight).
Indeed, lower bounds for $\CP^*$ (low-weight Cutting Planes) \cite{BonetPR97} were first established, and led to (unrestricted) CP lower bounds \cite{Pudlak97}.

Our second lower bound is a new linear depth lower bound for {\it semantic} Cutting Planes proofs.
(In a semantic Cutting Planes proof the deduction rules for CP are replaced by
a simple and much stronger \emph{semantic deduction rule}).

\begin{thm}
\label{depth-lb}
   For all sufficiently large $n$ there is a graph $G$ on $n$ vertices and a labelling $\ell$ such that the Tseitin formula for $(G, \ell)$ requires $\Omega(n)$ depth to refute in Semantic Cutting Planes.
\end{thm}

We note that depth lower bounds for Semantic Cutting Planes have already established
via communication complexity arguments. However, since Tseitin formulas have short communication protocols,
our depth bound for semantic Cutting Planes proofs of Tseitin is new.

\autoref{depth-lb} is established via a new technique for proving lower bounds on the depth of semantic Cutting Planes proofs. Our technique is inspired by the result of Buresh-Oppenheim et al.~\cite{BGHMP06}, who proved lower bounds on the depth of Cutting Planes refutations of Tseitin by studying the \emph{Ch\'{a}tal rank} of the associated polytope $P$. Letting $P^{(d)}$ be the polytope composed of all inequalities which can be derived in depth $d$ in Cutting Planes. The Ch\'{a}tal rank of $P$ is the minimum $d$ such that $P^{(d)}=\emptyset$. Thus, in order to establish a depth lower bound of depth $d$, one would like to show the existence of a point $p \in P^{(d)}$. To do so, they give a sufficient criterion for a point $p$ to be in $P^{(i)}$ in terms of the points in $P^{(i-1)}$. This criterion relies on a careful analysis of the specific rules of Cutting Planes, and is no longer sufficient for semantic $\CP$. Instead, we develop an analogous criterion for semantic $\CP$ by using novel \emph{geometric} argument (\autoref{lem:CruxForPM1}) which we believe will be of independent interest.

Our main motivation behind this depth bound is as a step towards proving a {\it supercritical} tradeoff in $\CP$ for Tseitin formulas.
A supercritical tradeoff for $\CP$, roughly speaking, states that small size $\CP$ proofs must sometimes necessarily be very deep --- that is, beyond the trivial depth upper bound of $O(n)$ \cite{Razborov16,BerkholzN20}.
(Observe that Dadush and Tiwari's quasipolynomial-size $\CP$ refutations of Tseitin are quasipolynomially deep; this is preserved by our simulation of Facelike Stabbing Planes by Cutting Planes in \autoref{thm:main}.)
Establishing supercritical tradeoffs is a major challenge, both because
hard examples witnessing such a tradeoff are rare, and because current methods seem to fail beyond the critical regime.
In fact, to date the only supercritical tradeoffs between size and depth for known proof systems are due to Razborov, under the additional assumption that the proofs have \emph{bounded width}.
Namely, Razborov exhibited a supercritical size-depth tradeoff for bounded width tree-like resolution \cite{Razborov16}, and then extended this result to $\CP$ proofs in which each inequality has a bounded number of distinct variables \cite{Razborov17}.

How could one prove a supercritical depth lower bound for Cutting Planes?
All prior depth lower bounds for Cutting Planes proceed by either reducing to communication complexity, or by using so-called \emph{protection lemmas} (e.g.~\cite{BGHMP06}).
Since communication complexity is always at most $n$, it will be useless for proving supercritical lower bounds directly.
It therefore stands to reason that we should focus on improving the known lower bounds using protection lemmas and, indeed, our proof of \autoref{depth-lb} is a novel geometric argument which generalizes the top-down ``protection lemma'' approach \cite{BGHMP06} for syntactic $\CP$.
At this point in time we are currently unable to use protection lemma techniques to prove size-depth tradeoffs, so, we leave this as an open problem.

\begin{conj}\label{conj:deepProofs}
	There exists a family of unsatisfiable formulas $\{F_n\}$ such that $F_n$ has quasipolynomial-size $\CP$ proofs, but any quasipolynomial-size proof requires superlinear depth.
\end{conj}

\subsection{Related Work}

\paragraph{Lower Bounds on $\SP$ and $\tRCP$.} Several lower bounds on subsystems of $\SP$ and $\tRCP$ have already been established. 
Kraj\'{i}\v{c}ek \cite{Krajicek98} proved exponential lower bounds on the size of $\RCP$ proofs in which both the \emph{width} of the clauses and the magnitude of the coefficients of each line in the proof are bounded. Concretely, let these bounds be $w$ and $c$ respectively. The lower bound that he obtains is $2^{n^{\Omega(1)}}/c^{w \log^2 n}$. Kojevnikov \cite{Kojevnikov07} removed the dependence on the coefficient size for $\tRCP$ proofs, obtaining a bound of $\exp(\Omega(\sqrt{n/w \log n}))$. Beame et al.~\cite{BeameFIKPPR18} provide a size-preserving simulation of Stabbing Planes by $\tRCP$ which translates a depth $d$ Stabbing Planes proof into a width $d$ $\tRCP$ proof, 
and therefore this implies lower bounds on the size of $\SP$ proofs of depth $o(n/\log n)$. Beame et al. \cite{BeameFIKPPR18} exhibit a function for which there are no $\SP$ refutations of depth $o(n / \log^2 n)$ via a reduction to the communication complexity of the CNF search problem. 

\paragraph{Supercritical Tradeoffs.} Besides the work of Razborov \cite{Razborov16}, a number of supercritical tradeoffs have been observed in proof complexity.
Perhaps most relevant for our work, Razborov \cite{Razborov17} proved a supercritical tradeoff for Cutting Planes proofs under the assumption that each inequality has a bounded number of distinct variables (mimicking the bound on the width of each clause in the supercritical tradeoff of \cite{Razborov16}). 

A number of supercritical tradeoffs are also known between proof width and proof \emph{space}. 
Beame et al.~\cite{BeameBI12} and Beck et al.~\cite{BeckNT13} exhibited formulas which admit polynomial size refutations in Resolution and the Polynomial Calculus respectively, and such that any refutation of sub-linear space necessitates a superpolynomial blow-up in size. Recently, Berkholz and Nordstr\"{o}m \cite{BerkholzN20} gave a supercritical trade-off between width and space for Resolution.

\paragraph{Depth in Cutting Planes and Stabbing Planes.} It is widely known (and easy to prove) that any unsatisfiable family of CNF formulas can be refuted by exponential size and \emph{linear} depth Cutting Planes.
It is also known that neither Cutting Planes nor Stabbing Planes can be \emph{balanced}, in the sense that a depth-$d$ proof can always be transformed into a size $2^{O(d)}$ proof \cite{BeameFIKPPR18, BGHMP06}.
This differentiates both of these proof systems from more powerful proof systems like Frege, for which it is well-known how to balance arbitrary proofs \cite{CookR79}.
Furthermore, even though both the Tseitin principles and systems of linear equations in finite fields can be proved in both quasipolynomial-size \emph{and} $O(\log^2 n)$ depth in Facelike $\SP$, the simulation of Facelike $\SP$ by $\CP$ \emph{cannot} preserve both size and depth, as the Tseitin principles are known to require depth $\Theta(n)$ to refute in $\CP$ \cite{BGHMP06}.

We first recall the known depth lower bound techniques for Cutting Planes, semantic Cutting Planes, and Stabbing Planes proofs. 
In all of these proof systems, arguably the primary method for proving depth lower bounds is by reducing to \emph{real communication complexity} \cite{ipu:cp ,BeameFIKPPR18}; however, communication complexity is always trivially upper bounded by $n$, and it is far from clear how to use the assumption on the size of the proof to boost this to superlinear.

A second class of methods have been developed for \emph{syntactic} Cutting Planes, which lower bound \emph{rank measures} of a polytope, such as the Chv\'{a}tal rank. In this setting, lower bounds are typically proven using so-called \emph{protection lemmas} \cite{BGHMP06}, which seems much more amenable to applying a small-size assumption on the proof. We also remark that for many formulas (such as the Tseitin formulas!) it is known how to achieve $\Omega(n)$-depth lower bounds in Cutting Planes via protection lemmas, while proving even $\omega(\log n)$ lower bounds via communication complexity is impossible, due to a known folklore upper bound. 

The first lower bound on the Chv\'{a}tal rank was established by Chv\'{a}tal et al.~\cite{ChvatalCH89}, who proved a linear bound for a number of polytopes in $[0,1]^n$. Much later, Pokutta and Schulz~\cite{PokuttaS11a} characterized the polytopes $P \subseteq [0,1]$ with $P \cap \mathbb{Z}^n = \emptyset$ which have Chv\'{a}tal rank exactly $n$. However, unlike most other cutting planes procedures, the Chv\'{a}tal rank is not of polytopes $P \cap [0,1]^n$ with $P \cap \mathbb{Z}^n = \emptyset$ is not upper bounded by $n$. Eisenbrand and Schulz~\cite{EisenbrandS99} showed that the Chv\'{a}tal rank of any polytope $P \subseteq [0,1]^n$ is at most $O(n^2 \log n)$ and gave examples where it is $\Omega(n)$; a nearly-matching quadratic lower bound was later established by Rothvo{\ss} and Sanita~\cite{RothvossS13}.
For CNF formulas, the Chv\'{a}tal rank is (trivially) at most $n$. Buresh-Oppenheim et al.~\cite{BGHMP06} gave the first lower bounds on the Chv\'{a}tal rank  a number of CNF formulas, including an $\Omega(n)$ lower bound for the Tseitin formulas. 

The rank of a number of generalizations of Cutting Planes has been studied as well. However, none of these appear to capture the strength of semantic Cutting Planes. Indeed, semantic Cutting Planes is able to refute Knapsack in a single cut, and therefore is known not to be polynomially verifiable unless $\P = \NP$ \cite{FilmusHL16}. Lower bounds on the rank when using split cuts and mixed integer cuts, instead of CG cuts, was established in \cite{CornuejolsL02a}.  
Pokutta and Schulz~\cite{PokuttaS10} obtained $\Omega(n/\log n)$ rank lower bounds on the complete tautology (which includes every clause of width $n$) for the broad class of \emph{admissible cutting planes}, which includes syntactic Cutting Planes, split cuts, and many of the lift-and-project operators. Bodur et al.~\cite{BodurPDMP18} studied the relationship between rank and integrality gaps for another broad generalization of Cutting Planes known as \emph{aggregate cuts}.

\section{Preliminaries} 
\label{sec:prelims}

We first recall the definitions of some key proof systems.

\paragraph{Resolution.} Fix an unsatisfiable CNF formula $F$ over variables $x_1, \ldots, x_n$. 
A \emph{Resolution refutation} $P$ of $F$ is a sequence of clauses $\set{C_i}_{i \in [s]}$ ending in the empty clause $C_s = \emptyset$ such that each $C_i$ is in $F$ or is derived from earlier clauses $C_j, C_k$ with $j, k < i$ using one of the following rules:
\begin{itemize}
	\item \emph{Resolution.} $C_i = (C_j \setminus \set{\ell_k}) \cup (C_k \setminus \set{\overline \ell_k})$ where $\ell_k \in C_j$, $\overline \ell_k \in C_k$ is a literal.
	\item \emph{Weakening.} $C_i \supseteq C_j$.
\end{itemize}
The \emph{size} of the resolution proof is $s$, the number of clauses.
It is useful to visualize the refutation $P$ as a directed acyclic graph; with this in mind the \emph{depth} of the proof (denoted $\depth_{\Res}(P)$) is the length of the longest path in the proof DAG.
The \emph{resolution depth} $\depth_{\Res}(F)$ of $F$ is the minimal depth of any resolution refutation of $F$.

\paragraph{Cutting Planes and Semantic Cutting Planes.} A \emph{Cutting Planes} ($\CP$) \emph{proof} of an inequality $cx \geq d$ from a system of linear inequalities $P$ is given by a sequence of inequalities \[ a_1x \geq b_1, a_2x \geq b_2, \ldots, a_s x \geq b_s \]
such that $a_s = c$, $b_s = d$, and each inequality $a_ix \geq b_i$ is either in $P$ or is deduced from earlier inequalities in the sequence by applying one of the two rules \emph{Linear Combination} or \emph{Division Rule} described at the beginning of Section \ref{sec:intro}.
We will usually be interested in the case that the list of inequalities $P$ defines a polytope.

An alternative characterization of Cutting Planes uses \emph{Chv\'{a}tal-Gomory cuts} (or just \emph{CG cuts}) \cite{CookCT87, Chvatal73a}.
Let $P$ be a polytope.
A hyperplane $ax = b$ is \emph{supporting} for $P$ if $b = \max \set{ax : x \in P}$, and if $ax = b$ is a supporting hyperplane then the set $P \cap \set{x \in \reals^n : ax = b}$ is called a \emph{face} of $P$.
An inequality $ax \leq b$ is \emph{valid} for $P$ if every point of $P$ satisfies the inequality and $ax = b$ is a supporting hyperplane of $P$.
\begin{defn}\label{def:cg-cut}
	Let $P \subseteq \reals^n$ be a polytope, and let $ax \geq b$ be any valid inequality for $P$ such that all coefficients of $a$ are relatively prime integers.
	The halfspace $\set{x \in \reals^n : ax \geq \lceil b \rceil}$ is called a \emph{CG cut} for $P$. (We will sometimes abuse notation and refer to the inequality $ax \geq \lceil b \rceil$ also as a CG cut.)
\end{defn}
If $ax \geq \lceil b \rceil$ is a CG cut for the polytope $P$, then we can derive $ax \geq \lceil b \rceil$ from $P$ in $O(n)$ steps of Cutting Planes by Farkas Lemma (note that the inequality $ax \geq b$ is valid for $P$ by definition, so we can deduce $ax \geq b$ as a linear combination of the inequalities of $P$ and then apply the division rule).
If $P$ is a polytope and $H$ is a CG cut, then we will write $P \vdash P \cap H$, and say that $P \cap H$ is \emph{derived} from $P$.

Given a CNF formula $F$, we can translate $F$ into a system of linear inequalities in the following natural way.
First, for each variable $x_i$ in $F$ add the inequality $0 \leq x_i \leq 1$.
If $C = \bigvee_{i \in P} x_i \vee \bigvee_{i \in N} \neg x_i$ is a clause in $F$, then we add the inequality \[\sum_{i \in P} x_i + \sum_{i \in N} (1-x_i) \geq 1. \]
It is straightforward to see that the resulting system of inequalities will have no integral solutions if and only if the original formula $F$ is unsatisfiable.
With this translation we consider Cutting Planes refutations (defined in the introduction) of $F$ to be refutations of the translation of $F$ to linear inequalities.

The \emph{semantic Cutting Planes} proof system (denoted $\sCP$ or Semantic $\CP$) is a strengthening of Cutting Planes proofs to allow \emph{any deduction} that is sound over Boolean points \cite{BonetPR97}.
Like Cutting Planes, an $\sCP$ proof is given by a sequence of halfspaces $\set{a_ix \geq c_i}_{i \in [s]}$, but now we can use the following very powerful \emph{semantic deduction rule}:
\begin{itemize}
	\item \emph{Semantic Deduction.} From $a_j x \geq c_j$ and $a_k x \geq c_k$ deduce $a_i x \geq c_i$ if every $\set{0,1}$ assignment satisfying both $a_jx \geq c_j$ and $a_k x \geq c_k$  also satisfies $a_ix \geq c_i$ .
\end{itemize}
Filmus et al.~\cite{FilmusHL16} showed that $\sCP$ is extremely strong: there are instances for which any refutation in $\CP$ requires exponential size, and yet these instances admit  polynomial-size refutations in semantic $\sCP$. 

The size of a Cutting Planes proof is the number of lines (it is known that for unsatisfiable CNF formulas that this measure is polynomially related to the length of the bit-encoding of the proof \cite{CookCT87}).
As with Resolution, it is natural to arrange Cutting Planes proofs into a proof DAG.
With this in mind we analogously define $\depth_{\CP}(F)$ and $\depth_{\sCP}(F)$ to be the smallest depth of any (semantic) Cutting Planes proof of $F$. 

It is known that \emph{any} system of linear inequalities in the unit cube has $\CP$ depth at most
$O(n^2 \log n)$, and moreover there
are examples requiring $\CP$-depth more than $n$ \cite{EisenbrandS99}.
However for unsatisfiable CNF formulas,
the $\CP$-depth is at most $n$ \cite{BockmayrEHS99}.

\paragraph{Stabbing Planes.}
Let $F$ be an unsatisfiable system of linear inequalities.
A \emph{Stabbing Planes ($\SP$) refutation} of $F$ is a directed binary tree, $T$, where each edge is labelled with a linear integral inequality satisfying the following \emph{consistency conditions}:
\begin{itemize}
	\item \emph{Internal Nodes.} For any internal node $u$ of $T$, if the right outgoing edge of $u$ is labelled with $ax \geq b$, then the left outgoing edge is labelled with its \emph{integer negation} $ax \leq  b-1 $.
	\item \emph{Leaves.} Each leaf node $v$ of $T$ is labelled with a non-negative linear combination of inequalities in $F$ with inequalities along the path leading to $v$ that yields $0 \geq 1$.
\end{itemize}
For an internal node $u$ of $T$, the pair of inequalities $(ax \le b-1, ax \ge b$) is called the \emph{query} corresponding to the node. Every node of $T$ has a polytope $P$ associated with it, where $P$ is the polytope defined by the intersection of the inequalities in $F$ together with the inequalities labelling the path from the root to this node. We will say that the polytope $P$ \emph{corresponds} to this node. 
The \emph{slab} corresponding to the query is $\{ x^* \in \mathbb{R}^n \mid b-1 < ax^* < b \}$, which is the set of points ruled out by this query. The \emph{width} of the slab is the minimum distance between $ax \leq b-1$ and $ax \geq b$, which is $1/\|a\|_2$.
The \emph{size} of a refutation is the bit-length needed to encode a description of the entire proof tree, which, for CNF formulas as well as sufficiently bounded systems of inequalities, is polynomially equivalent to the number of queries  in the refutation \cite{TseitinUpperBound}.
As well, the \emph{depth} of the refutation is the depth of the binary tree.
The proof system $\SP^*$ is the subsystem of Stabbing Planes obtained by restricting all coefficients of the proofs to have magnitude at most quasipolynomial ($n^{\log^{O(1)} n}$) in the number of input variables.

The Stabbing Planes proof system was introduced by Beame et al.~\cite{BeameFIKPPR18} as a generalization of Cutting Planes that more closely modelled query algorithms and branch-and-bound solvers.
Beame et al.~proved that $\SP$ is equivalent to the proof system $\mathsf{TreeR}(\CP)$ introduced by Kraj\'{i}\v{c}ek \cite{Krajicek98} which can be thought of as a generalization of Resolution where the literals are replaced with integer-linear inequalities.

\section{Translating Stabbing Planes into Cutting Planes} 
\label{sec:SP}

\subsection{Equivalence of $\CP$ with Subsystems of $\SP$}
\label{sec:facelikesp-sim}

In this section we prove \autoref{thm:main}, restated below, which characterizes Cutting Planes as a non-trivial subsystem of Stabbing Planes.

\main*

We begin by formally defining Facelike $\SP$.
\begin{defn}
	A Stabbing Planes query $(ax \leq b-1, ax \geq b)$ at a node $P$ is \emph{facelike} if one of the sets $P \cap \{x \in \mathbb{R}^n: ax \leq b-1\}$, $P \cap \{x \in \mathbb{R}^n: ax \geq b\}$ is empty or a face of $P$ (see \autoref{fig:fSPqueries}).
	An $\SP$ refutation is facelike if every query in the refutation is facelike. 
\end{defn}

Enroute to proving \autoref{thm:main}, it will be convenient to introduce the following further restriction of Facelike Stabbing Planes.

\begin{defn}
	A Stabbing Planes query $(ax \leq b-1, ax \geq b)$ at a node corresponding to a polytope $P$ is \emph{pathlike} if at least one of $P \cap \{x \in \mathbb{R}^n: ax \leq b-1\}$ and $P \cap \{x \in \mathbb{R}^n: ax \geq b\}$ is empty (see \autoref{fig:pSPqueries}). A Pathlike $\SP$ refutation is one in which every query is pathlike. 
\end{defn}

The name ``pathlike'' stems from the fact that the underlying graph of a pathlike Stabbing Planes proof is a path, since at most one child of every node has any children (see Figure \ref{fig:pathlike-vs-facelike}).
In fact, we have already seen (nontrivial) pathlike $\SP$ queries under another name: Chv\'{a}tal-Gomory cuts.

\begin{figure}[htb]
	\centering
	\begin{subfigure}[b]{1\textwidth}
	\begin{tikzpicture}[scale=0.6]
		\draw[thick, ->] (-2,5)  -- (-0.5,3.5);
		\draw[thick, ->] (-2,5)  -- (-3.5,3.5);
		\filldraw[fill=green!9,thick] (-2,5) circle (8pt);
		\node[text width=2cm] at (1,4.4) { $ax \geq b$};
		\node[text width=2cm] at (-4.5,4.4) {$ax \leq b-1$};
		\filldraw[fill=red!15,thick,  color =red!25] (-3.75,3.2) circle (10pt);
		\node[text width=1cm] at (-3.07,3.22) { $\emptyset$};
		\draw[thick, ->,dashed] (-0.4,3.35)  -- (1.1,1.85);
		\draw[thick, ->] (-0.4,3.35)  -- (-1.9,1.85);
		\filldraw[fill=green!9,thick] (-0.35,3.25) circle (8pt);
		\filldraw[fill=red!15,  color =red!25] (-2.05,1.5) circle (10pt);
		\node[text width=1cm] at (-1.38,1.52) { $\emptyset$};
		
		
		
		
		\node[text width=2cm] at (-4.6,0) {\phantom{a}};	
	\end{tikzpicture} \hspace{1.3em}
	\begin{tikzpicture}[scale=0.6]
	\filldraw[color=black!60, fill=green!9, very thick] (6.4,-0.5)-- (3.8,0.7) -- (3.9,3.2) -- (6.3,4.5) -- (8.8,3.2) -- (8.7,0.5) -- cycle;
	\filldraw[color=white, fill=white!, very thick,opacity=0.7] (3.1,5.3)  -- (3.86,-1.5)  -- (4.6,-1.4) -- (4.55,5.4) -- cycle;
	\draw[color=black!60, very thick] (3.4,5.3)  -- (3.4,-1.5);
	\draw[color=black!60, very thick,->] (4.6,2)  -- (4.9,2);
	\draw[color=black!60, very thick] (4.6,5.4)  -- (4.6,-1.4);
	\draw[color=black!60, very thick,->] (3.4,2)  -- (3.1,2);
	\node[text width=3cm] at (2.5,2) {$ax \leq b-1$};
	\node[text width=3cm] at (7.7,2) {$ax \geq b$};
	\node[text width=1cm] at (6.7,4) {$P$};
	\end{tikzpicture}
	\caption{A Pathlike query. The polytope $P \cap \{x \in \mathbb{R}^n :ax \leq b-1\} = \emptyset$, and $ax \geq b$ is a CG cut for $P$. }\label{fig:pSPqueries}
	\end{subfigure}

	\begin{subfigure}[b]{1\textwidth}
	\begin{tikzpicture}[scale=0.6]
		\draw[thick, ->] (-1,5)  -- (1.5,3.5);
		\draw[thick, ->] (-1,5)  -- (-3.5,3.5);
		\filldraw[fill=white,thick] (-1,5) circle (8pt);
		\node[text width=2cm] at (2.4,4.5) { $ax \geq b$};
		\node[text width=2cm] at (-3.9,4.5) {$ax \leq b-1$};
		\draw[thick, ->, dashed] (1.65,3.35)  -- (2.65,2.35);
		\draw[thick, ->,dashed] (1.65,3.35)  -- (0.65,2.35);
		\filldraw[fill=green!9] (1.65,3.2) circle (8pt);
		\draw[thick, ->, dashed] (-3.65,3.35)  -- (-2.65,2.35);
		\draw[thick, ->,dashed] (-3.65,3.35)  -- (-4.65,2.35);
		\filldraw[fill=orange!40,thick,thick] (-3.65,3.25) circle (8pt);

		\node[text width=2cm] at (-4.6,0) {\phantom{a}};	
	\end{tikzpicture}
	\begin{tikzpicture}[scale=0.6]
	\filldraw[color=black!60, fill=green!9, very thick] (6.4,-0.5)-- (3.8,0.7) -- (3.9,3.2) -- (6.3,4.5) -- (8.8,3.2) -- (8.7,0.5) -- cycle;
	\filldraw[color=white, fill=white!, very thick,opacity=0.7] (3.1,5.3)  -- (3.86,-1.5)  -- (4.7,-1.4) -- (4.9,5.4) -- cycle;
	\draw[color=black!60, very thick] (3.95,5.3)  -- (3.7,-1.5);
	\draw[color=black!60, very thick,->] (4.8,2)  -- (5.15,1.98);
	\draw[color=black!60, very thick] (4.95,5.4)  -- (4.7,-1.4);
	\draw[color=black!60, very thick,->] (3.8,2.01)  -- (3.48,2.029);
	\draw[color=orange!70, very thick] (3.78,0.7) -- (3.88,3.2);
	\node[text width=3cm] at (3,2) {$ax \leq b-1$};
	\node[text width=3cm] at (7.8,2) {$ax \geq b$};
	\node[text width=1cm] at (6.7,4) {$P$};
	\node[text width=3cm] at (3.8,-1.8) {$ax =b-1$};
	\end{tikzpicture}
	\caption{A Facelike query. The polytope $P \cap \{x \in \mathbb{R}^n:ax\leq b-1\} = P \cap \{x \in \mathbb{R}^n:ax=b-1\}$ is a face of $P$.} \label{fig:fSPqueries}
	\end{subfigure}
	\caption{Pathlike and Facelike $\SP$ queries on a polytope $P$. On the left are the proofs and on the right are the corresponding effects on the polytope.}
	\label{fig:pathlike-vs-facelike}
\end{figure}
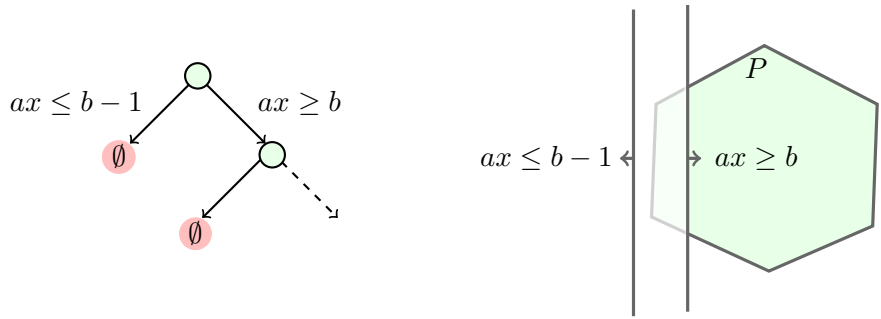
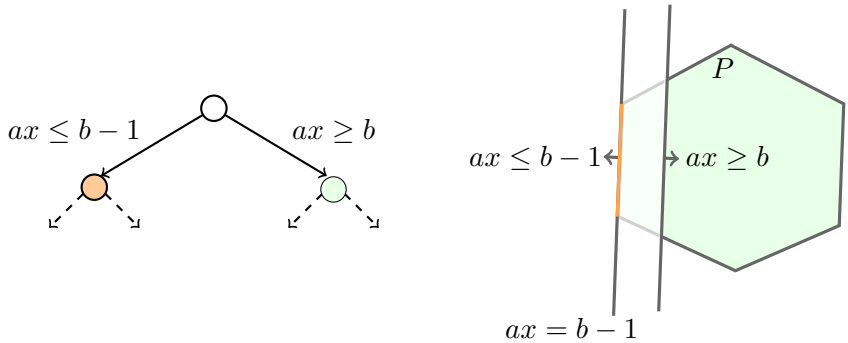 

\begin{lem}\label{lem:cg-cut}
	Let $P$ be a polytope and let $(ax \leq b-1, ax \geq b)$ be a pathlike Stabbing Planes query for $P$.
	Assume w.l.o.g. that $P \cap \set{x \in \reals^n : ax \leq b-1} = \emptyset$ and that $P \cap \set{x \in \reals^n : ax \geq b} \subsetneq P$.
	Then $ax \geq b$ is a CG cut for $P$.
\end{lem}
\begin{proof}
	Since $ax \geq b$ is falsified by some point in $P$, it follows that there exists some $0 < \varepsilon < 1$ such that $ax \geq b - \varepsilon$ is valid for $P$ --- note that $\varepsilon < 1$ since otherwise $ax \leq b-1$ would not have empty intersection with $P$.
	This immediately implies that $ax \geq b$ is a CG cut for $P$.
\end{proof}

With this observation we can easily prove that Pathlike $\SP$ is equivalent to $\CP$. Throughout the remainder of the section, for readability, we will use the abbreviation $P \cap \{ax \geq b\}$ for $P \cap \{x \in \mathbb{R}^n : ax \geq b\}$, for any polytope $P$ and linear inequality $ax \geq b$.

\begin{lem}\label{lem:pathsp-cp}
	 Pathlike $\SP$ is polynomially equivalent to $\CP$. 	
\end{lem}
\begin{proof}
	First, let $a_1x \geq b_1, a_2x \geq b_2, \ldots, a_sx \geq b_s$ be a $\CP$ refutation of an unsatisfiable system of linear inequalities $Ax \geq b$. 
	Consider the sequence of polytopes $P_0 =\{Ax \geq b\}$ and $P_i = P_{i-1} \cap \{a_i x \geq b_i\}$. 
	By inspecting the rules of $\CP$, it can observed that $P_i \cap \{a_ix \leq b_i-1\}= \emptyset$ and thus $P_{i+1}$ can be deduced using one pathlike $\SP$ query from $P_i$ for all $0 \leq i \leq s$.
	
	Conversely, let $P$ be any polytope and let $(ax \leq b-1, ax \geq b)$ be any pathlike $\SP$ query to $P$ (so, suppose w.l.o.g.~that the halfspace defined by $ax \leq b-1$ has empty intersection with $P$).
	By \autoref{lem:cg-cut}, $ax \geq b$ is a CG cut for $P$, and so can be deduced in Cutting Planes from the inequalities defining $P$ in length $O(n)$ (cf.~Section \ref{sec:prelims}).
	Applying this to each query in the Pathlike $\SP$ proof yields the theorem.

\end{proof}
	Next, we show how to simulate Facelike $\SP$ proofs by Pathlike $\SP$ proofs of comparable size.
	The proof of \autoref{lem:SPtopSP} is inspired by Dadush and Tiwari \cite{TseitinUpperBound}, and will use the following lemma due to Schrijver~\cite{SCHRIJVER1980291} (although, we use the form appearing in \cite{CookCT87}).
	Recall that we write $P \vdash P'$ for polytopes $P, P'$ to mean that $P'$ can be obtained from $P$ by adding a single CG cut to $P$.
	\begin{lem}[Lemma 2 in \cite{CookCT87}]
	\label{lem:shrijver}
		Let $P$ be a polytope defined by a system of integer linear inequalities and let $F$ be a face of $P$. 
		If $F \vdash F'$ then there is a polytope $P'$ such that $P \vdash P'$ and $P' \cap F \subseteq F'$.
	\end{lem}
	%
	
	\begin{lem}
	\label{lem:SPtopSP}
		Facelike $\SP$ is polynomially equivalent to Pathlike $\SP$.	
	\end{lem}
	\begin{proof}
		That Facelike $\SP$ simulates Pathlike $\SP$ follows by the fact that any Pathlike $\SP$ query is a valid query in Facelike $\SP$.  
		For the other direction, consider an $\SP$ refutation $\pi$ of size $t$. 
		We describe a recursive algorithm for generating a Pathlike $\SP$ proof from $\pi$.
		The next claim will enable our recursive case.
	
	\spcnoindent
	{\bf Claim.} Let $P$ be a polytope and suppose $ax \geq b$ is valid for $P$. 
	Assume that $P \cap \{ax = b\}$ has a Pathlike $\SP$ refutation using $s$ queries. 
	Then $P \cap \{ax \geq b+1\}$ can be derived from $P$ in Pathlike $\SP$ using $s+1$ queries.
	\begin{proof}[Proof of Claim.]
		Since $ax \geq b$ is valid for $P$ it follows that $F = P \cap \{ax = b\}$ is a face of $P$ by definition.
		Consider the Pathlike $\SP$ refutation  $F_0, F_1, \ldots F_s = \emptyset$, where the $i$th polytope $F_i$ for $i < s$ is obtained from $F_{i-1}$ by applying a pathlike $\SP$ query and proceeding to the non-empty child.
		Without loss of generality we may assume that $F_i \subsetneq F_{i-1}$ for all $i$, and so applying \autoref{lem:cg-cut} we have that $F_{i-1} \vdash F_i$ for all $i$.
		Thus, by applying \autoref{lem:shrijver} repeatedly, we get a sequence of polytopes $P = P_0 \vdash P_1 \vdash \cdots \vdash P_s$ such that $P_i \cap F = P_i \cap \{ax = b\} \subseteq F_i$. 
		This means that $P_s \cap \{ax = b\} \subseteq F_s = \emptyset$, and so $(ax \leq b, ax \geq b+1)$ is Pathlike $\SP$ query for $P_s$. 
		This means that $P_s \vdash P_s \cap \{ax \geq b+1\} \subseteq P \cap \{ax \geq b+1\}$.
		Since any CG cut can be implemented as a Pathlike $\SP$ query the claim follows by applying the $s$ CG cuts as pathlike queries, followed by the query $(ax \leq b, ax \geq b+1)$.
	\end{proof}
	
	We generate a Pathlike $\SP$ refutation by the following recursive algorithm, which performs an \emph{in-order} traversal of $\pi$. 
	At each step of the recursion (corresponding to a node in $\pi$) we maintain the current polytope $P$ we are visiting and a Pathlike $\SP$ proof $\Pi$ --- initially, $P$ is the initial polytope and $\Pi = \emptyset$. 
	We maintain the invariant that when we finish the recursive step at node $P$, the Pathlike $\SP$ refutation $\Pi$ is a refutation of $P$. 
	The algorithm is described next:
	\begin{enumerate}
		\item Let $(ax \leq b-1,~ ax \geq b)$ be the current query and suppose that $ax \geq b-1$ is valid for $P$.
		\item Recursively refute $P \cap \{ ax \leq b-1\} = P \cap \set{ax = b-1}$, obtaining a Pathlike $\SP$ refutation $\Pi$ with $t$ queries. 
		\item Apply the above Claim to deduce $P \cap \{ax \geq b\}$ from $P$ in $t+1$ queries. 
		\item Refute $P \cap \{ ax \geq b\}$ by using the $\SP$ refutation for the right child. 
	\end{enumerate}
	Correctness follows immediately from the Claim, and also since the size of the resulting proof is the same as the size of the $\SP$ refutation. 		
	\end{proof}
	
	\autoref{thm:main} then follows by combining \autoref{lem:pathsp-cp} with \autoref{lem:SPtopSP}.

\subsection{Simulating $\SP^*$ by $\CP$}
\label{sec:quasipoly-simulation}
In this section we prove \autoref{thm:quasipoly-simulation}, restated below for convenience.

\quasi*

To prove this theorem, we will show that \emph{any} low coefficient $\SP$ proof can be converted into a Facelike $\SP$ proof with only a quasi-polynomial loss. 
If $P$ is a polytope let $d(P)$ denote the \emph{diameter} of $P$, which is the maximum Euclidean distance between any two points in $P$.
\autoref{thm:quasipoly-simulation} follows immediately from the following theorem.
\begin{thm}\label{thm:low-coefficients}
	Let $P$ be a polytope and suppose there is an $\SP$ refutation of $P$ with size $s$ and maximum coefficient size $c$. 
	Then there is a Facelike $\SP$ refutation of $P$ in size 
	\[ s (c \cdot d(P) \sqrt{n})^{\log s}. \]	
\end{thm}
\begin{proof}
	The theorem is by induction on $s$. Clearly, if $s=1$ then the tree is a single leaf and the theorem is vacuously true.
	
	We proceed to the induction step. Let $P$ be the initial polytope and $\pi$ be the $\SP$ proof. Consider the first query $(ax \leq b, ~ax \geq b+1)$ made by the proof, and let $\pi_L$ be the $\SP$ proof rooted at the left child (corresponding to $ax \leq b$) and let $\pi_R$ be the $\SP$ proof rooted at the right child. Let $P_L$ denote the polytope at the left child and $P_R$ denote the polytope at the right child. By induction, let $\pi'_L$ and $\pi'_R$ be the Facelike $\SP$ refutations for $P_L$ and $P_R$ guaranteed by the statement of the theorem.
	
	Suppose w.l.o.g. that $|\pi_L| \leq |\pi|/2$. Let $b_0$ be the largest integer such that $ax \geq b_0$ is satisfied for any point in $P$. The plan is to replace the first query $(ax \leq b, ~ax \geq b+1)$ with a sequence of queries $q_0, q_1, \ldots, q_{t-1}$ such that
	\begin{itemize}
		\item For each $i$, $q_i =(ax \leq b_0+i,~ax \geq b_0+i+1)$.
		\item The query $q_0$ is the root of the tree and $q_i$ is attached to the right child of $q_{i-1}$ for $i \geq 1$.
		\item $q_{t-1} = (ax \leq b,~ax \geq b+1)$.
	\end{itemize}
	After doing this replacement, instead of having two child polytopes $P_L, P_R$ below the top query, we have $t+1$ polytopes $P_0,P_1,\ldots,P_{t+1}$ where $P_i = P \cap \{ax =b_0 +i\}$ and $P_{t+1} = P_R$. To finish the construction, for each $i \leq t$ use the proof $\pi'_L$ to refute $P_i$ and the proof $\pi'_R$ to refute $P_{t+1}$. 
	
	We need to prove three statements: this new proof is a valid refutation of $P$, the new proof is facelike, and that the size bound is satisfied. 
	
	First, it is easy to see that this is a valid proof, since for each $i \leq t$ the polytope $P_i \subseteq P_L$ and $P_{t+1} \subseteq P_R$ --- thus, the refutations $\pi_L'$ and $\pi'_R$ can be used to refute the respective polytopes. 
	
	Second, to see that the proof is facelike, first observe that all the queries in the subtrees $\pi'_L, \pi'_R$ are facelike queries by the inductive hypothesis. 
	So, we only need to verify that the new queries at the top of the proof are facelike queries, which can easily be shown by a quick induction. First, observe that the query $q_0$ is a facelike query, since $b_0$ was chosen so that $ax \geq b_0$ is valid for the polytope $P$. 
	By induction, the query $q_i = (ax \leq b_0 + i, ~ax \geq b_0 + i+1)$ is a facelike query since the polytope $P_i$ associated with that query is $P \cap \{ax \geq b_0 + i\}$ by definition. Thus $ax \geq b_0 + i$ is valid for the polytope at the query.
	
	Finally, we need to prove the size upper bound. Let $s$ be the size of the original proof, $s_L$ be the size of $\pi_L$ and $s_R$ be the size of $\pi_R$. Observe that the size of the new proof is given by the recurrence relation
	\[ f(s) = t \cdot f(s_L)+f(s_R). \] 
	where $f(1) = 1$. Since the queries $q_0,q_1,\ldots, q_{t-1}$ cover the polytope $P_L$ with slabs of width $1/\|a\|_2$, it follows that
	\[ t \leq d(P_L) \|a\|_2 \leq d(P) \sqrt{n} \|a\|_\infty =d(P)c \sqrt{n} \]
	where we have used that the maximum coefficient size in the proof is $c$. Thus, by induction, the previous inequality, and the assumption that $s_L \leq s/2$, we can conclude that the size of the proof is
	\begin{align*}
		f(s) &\leq s_L(c \cdot d(P) \sqrt{n})(c \cdot d(P_L) \sqrt{n})^{\log s_L} + s_R(c \cdot d(P_R) \sqrt{n} +)^{\log s_R} \\
		&\leq s_L(c \cdot d(P) \sqrt{n})(c \cdot d(P) \sqrt{n})^{\log (s/2)} + s_R(c \cdot d(P) \sqrt{n} )^{\log s}\\
		&\leq s_L(c \cdot d(P) \sqrt{n})^{\log s} + s_R(c \cdot d(P) \sqrt{n})^{\log s} \\
		&=s(c \cdot d(P) \sqrt{n})^{\log s}. \qedhere
	\end{align*}
\end{proof}

\autoref{thm:quasipoly-simulation} follows immediately, since for any CNF formula $F$ the encoding of $F$ as a system of linear inequalities is contained in the $n$-dimensional cube $[0,1]^n$, which has diameter $\sqrt{n}$. 
We may also immediately conclude \autoref{thm:spstar-lb} by applying the known lower bounds on the size of Cutting Planes proofs \cite{Pudlak97,FlemingPPR17,HrubesP17,GargGKS18}. 

As a consequence of \autoref{thm:quasipoly-simulation} and the non-automatability of Cutting Planes \cite{GoosKMP20}, we can conclude that $\SP^*$ proofs cannot be found efficiently assuming $\P \neq \NP$.
\begin{cor}
\label{cor:non-automatability} 
	$\SP^*$ is not automatable unless $\P \neq \NP$. 
\end{cor}
This follows by observing that the argument in \cite{GoosKMP20} does not require large coefficients.

\section{Refutations of Linear Equations over a Finite Field}\label{sec:qpSPref}

In this section we prove \autoref{thm:linear-equations}. 
To do so, we will extend the approach used by Beame et al.~\cite{BeameFIKPPR18} to prove quasi-polynomial upper bounds on the Tseitin formulas to work on any unsatisfiable set of linear equations over any finite field. 

If $ax = b$ is a linear equation we say the \emph{width} of the equation is the number of non-zero variables occurring in it.
Any width-$d$ linear equation over a finite field of size $q$, denoted $\mathbb{F}_q$, can be represented by a CNF formula with $q^{d-1}$ width-$d$ clauses --- one ruling out each falsifying assignment.
For a width-$d$ system of $m$ linear equations $F$ over $\mathbb{F}_q$, we will denote by $|F| := mq^{d-1}$ the size of the CNF formula encoding $F$. 

\begin{thm}\label{thm:SP_gaussian_elim}
	Let $F = \{f_1 = b_1, \ldots, f_m = b_m \}$ be a width-$d$, unsatisfiable set of linear equations over $\mathbb{F}_q$.
	There is an $\SP$ refutation of (the CNF encoding of) $F$ in size $(mqd)^{O(\log m)}q^d = |F|^{O(\log m)}$.
\end{thm}

First we sketch the idea for $\mathbb{F}_2$, i.e., a system of XOR equations.
In this case the $\SP$ proof corresponds to a branch decomposition procedure which is commonly used to solve SAT (see e.g. \cite{LodhaOS19,AlekhnovichR02,Dechter96,Darwiche01}). 
View the system $F$ as a hypergraph over $n$ vertices (corresponding to the variables) and with a $d$-edge for each equation. 
Partition the set of hyperedges into two sets $E = E_1 \cup E_2$ of roughly the same size, and consider the \emph{cut} of vertices that belong to both an edge in $E_1$ and in $E_2$. 
Using the $\SP$ rule we branch on all possible values of the sum of the cut variables in order to isolate $E_1$ and $E_2$. 
Once we know this sum, we are guaranteed that either $E_1$ is unsatisfiable or $E_2$ is unsatisfiable depending on the parity of the of the sum of the cut variables. 
This allows us to recursively continue on the side of the cut ($E_1$ or $E_2$) that is unsatisfiable. 
Since there are $n$ Boolean variables, each cut corresponds to at most $n+1$ possibilities for the sum, and if we maintain that the partition of the hyper edges defining the cut is balanced, then we will recurse at most $O(\log m)$ times. 
This gives rise to a tree decomposition of fanout $O(n)$ and height $O(\log n)$.

Over a finite field of size $q$ the proof will proceed in much the same way. 
Instead of a subgraph, at each step we will maintain a subset of the equations $I \subseteq [m]$ such that $\{f_i =b_i \}_{i \in I}$ must contain a constraint that is violated by the $\SP$ queries made so far. We partition $I$ into two sets $I_1$ and $I_2$ of roughly equal size and query the values $a$ and $b$ of $\sum_{i \in I_1} f_i$ and $\sum_{i \in I_2} f_i$. Because $F$ is unsatisfiable, at least one of $a- \sum_{i \in I_1} b_i \not \equiv 0$ or $b- \sum_{i \in I_2} b_i \not \equiv 0$, meaning that that it is unsatisfiable, and we recurse on it.


In the following, we will let $z$ stand for a vector of $\mathbb{F}_q$-valued variables $z_i$. When we discuss any form $f := az$ where $a \in \mathbb{F}_q^n$ and $z$ is a vector of $n$ variables $z_i$, we will implicitly associate it with the linear form $\sum_{i \in [n]} a_i (\sum_{j \in [\log q]} x_{i,j})$ where $x_{i,j}$ are the $\log q$ many Boolean variables encoding $z_i$ in the CNF encoding of $F$.

\begin{proof}[Proof of \autoref{thm:SP_gaussian_elim}]
	Let $F = \{f_1 = b_1,\ldots, f_m = b_m\}$ be a system of unsatisfiable linear equations over $\mathbb{F}_q$, where each $f_i = a_iz$ for $a_i \in \mathbb{F}_q^n$, and $b_i \in \mathbb{F}_q$. Because $F$ is unsatisfiable, there exists a  $\mathbb{F}_q$ linear combination of the equations in $F$ witnessing this; formally, there exists $\alpha \in \mathbb{F}_q^n$ such that $\sum_{i \in [m]}  \alpha_i f_i  \equiv 0 \mod q$, but $\sum_{i \in [m]} \alpha_i b_i \not \equiv 0 \mod q$. 
	
	Stabbing Planes will implement the following binary search procedure for a violated equation; we describe the procedure first, and then describe how to implement it in Stabbing Planes. In each round we maintain a subset $I \subseteq [m]$ and an integer $k_I$ representing the value of $\sum_{i \in I} \alpha_i f_i$. Over the algorithm, we maintain the invariant that  $k_I - \sum_{i \in I} b_i \not \equiv 0 \bmod q$, which implies that there must be a contradiction to $F$ inside of the constraints $\{f_i = b_i\}_{i \in I}$.
	
	Initially, $I = [m]$ and we obtain $k_I$ by querying the value of the sum $\sum_{i \in [m]} \alpha_i f_i$. If $k_I \not \equiv 0 \mod q$ then this contradicts the fact that $\sum_{i \in I} \alpha_i f_i \equiv 0 \mod q$; thus, the invariant holds.  Next, perform the following algorithm.
	\begin{enumerate}
		\item Choose a balanced partition $I = I_1 \cup I_2$ (so that $||I_1|-|I_2|| \leq 1$).
		\item Query the value of $\sum_{i \in I_1} \alpha_i f_i$ and $ \sum_{i \in I_2} \alpha_i f_i$; denote these values by $a$ and $b$ respectively.
		\item If $a -\sum_{i \in I_1} \alpha_i b_i \not \equiv 0 \bmod q$ then recurse on $I_1$ with $k_{I_1} := a$. Otherwise, if $b -\sum_{i \in I_2} \alpha_i b_i \not \equiv 0 \bmod q$ then recurse on $I_2$ with $k_{I_2} := b$.
		\item Otherwise (if $a -\sum_{i \in I_1} \alpha_i b_i \equiv b -\sum_{i \in I_2} \alpha_i b_i \equiv 0 \bmod q$), then this contradicts the invariant:
			\begin{align*} 0  \not \equiv k_I - \sum_{i \in I} \alpha b_i
			&= \sum_{i \in I} \alpha_i (f_i-b_i) \\ &= \sum_{i \in I_1} \alpha_i (f_i-b_i) + \sum_{i \in I_2} \alpha_i (f_i-b_i) \\ 
			&= (a-\sum_{i\in I_1} \alpha_ib_i) + (b-\sum_{i\in I_1} \alpha_ib_i) \equiv 0 \mod q.	
			\end{align*}
	\end{enumerate}
	This recursion stops when $|I| = 1$, at which point we have an immediate contradiction between $k_I$ and the single equation indexed by $I$. 
	
	It remains to implement this algorithm in $\SP$. First, we need to show how to perform the queries in step 2. 
	Querying the value of any sum $\sum_{i \in I} \alpha_i f_i$ can be done in a binary tree with at most $q^2md$ leaves, one corresponding to every possible query outcome. Internally, this tree queries all possible integer values for this sum (e.g. $(\sum_{i \in I} \alpha_i f_i \leq 0, \sum_{i \in I} \alpha_i f_i \geq 1), (\sum_{i \in I} \alpha_i f_i \leq 1, \sum_{i \in I} \alpha_i f_i \geq 2), \ldots$). For the leaf where we have deduced $\sum_{i \in [m]} \alpha_i f_i \leq 0$ we use the fact that each variable is non-negative to deduce that $\sum_{i \in [m]} \alpha_i f_i \geq 0$ as well. 
	Note that $q^2md$ is an upper bound on this sum because there are $m$ equations, each containing at most $d$ variables, each taking value at most $(q-1)$ \footnote{Note that instead of querying the value of $\sum_{i \in I} \alpha_i f_i$ we could have queried $\sum_{i \in I} \alpha_i f_i ~(\bmod q)$ to decrease the number of leaves to $qmd$.}. Thus, step 2 can be completed in $(q^2md)^2$ queries.
	
	Finally, we show how to derive refutations in the following cases: (i) when we deduced that $\sum_{i \in [m]} \alpha_i f_i \not \equiv 0 \mod q$ at the beginning, (ii) in step 4, (iii) when $|I| = 1$. 
	\begin{enumerate}
		\item[(i)]	Suppose that we received the value $a \not \equiv 0 \mod q$ from querying $\sum_{i \in [m]} \alpha_i f_i$. Note that every variable in $\sum_{i \in [m]} \alpha_i f_i$ is a multiple of $q$. Query 
	\[\Big(\sum_{i \in [m]} \alpha_i f_i/q \leq \lceil a/q \rceil-1,~ \sum_{i \in [m]} \alpha_i f_i/q \geq \lceil a/q \rceil \Big).\]
	At the leaf that deduces $\sum_{i \in [m]} \alpha_i f_i/q \leq \lceil a/q \rceil-1$, we can derive $0 \geq 1$ as a non-negative linear combination of this inequality together with $\sum_{i \in [m]} \alpha_i f_i \geq a$. Similarly, at the other leaf $\sum_{i \in [m]} \alpha_i f_i/q \geq \lceil a/q \rceil$ can be combined with $\sum_{i \in [m]} \alpha_i f_i \leq a$ to derive $0 \geq 1$.
	\item[(ii)] Suppose that $a -\sum_{i \in I_1} \alpha_i b_i \equiv b -\sum_{i \in I_2} \alpha_i b_i \equiv 0 \bmod q$. Then $0 \geq 1$ is derived by summing $\sum_{i \in I_1} \alpha_i f_i \geq a$, $\sum_{i \in I_2} \alpha_i f_i \geq b$ and $\sum_{i \in I} \alpha_i f_i \leq k_I$, all of which have already been deduced.
	\item[(iii)] When $|I| = 1$ then we deduced that $a_I z = k_I$ for $k_I \not \equiv b_I \mod q$ and we would like to derive a contradiction using the axioms encoding $a_I z \equiv b_I$. These axioms are presented to $\SP$ as the linear-inequality encoding of a CNF formula, and while there are no integer solutions satisfying both these axioms and $a_I z = k_I$, there could in fact be \emph{rational} solutions. To handle this, we simply force that each of the at most $d$ variables in $a_I z$ takes an integer value by querying the value of each variable one by one. As there are at most $d$ variables, each taking an integer value between $0$ and $q-1$, this can be done in a tree with at most $q^d$ many leaves. At each leaf of this tree we deduce $0 \geq 1$ by a non-negative linear combination with the axioms, the integer-valued variables, and $a_Iz \equiv b_I$. 
	\end{enumerate}
	The recursion terminates in at most $O(\log m)$ many rounds because the number of equations under consideration halves every time. Therefore, the size of this refutation is $(qmd)^{O(\log m)} q^d$. Note that by making each query in a balanced tree, this refutation can be carried out in depth $O(\log^2(mqd))$.
\end{proof}

Finally, we conclude \autoref{thm:linear-equations}.

\begin{proof}[Proof of \autoref{thm:linear-equations}]
Observe that the $\SP$ refutation from \autoref{thm:SP_gaussian_elim} is facelike. Indeed, to perform step 2 we query $(\sum_{i \in I} \alpha_i f_i \leq t-1, \sum_{i \in I} \alpha_i f_i \geq t)$ from $t = 1,\ldots, q^2md$. For $t=1$, the halfspace $\sum_{i \in I} \alpha_i f_i \geq 0$ is valid for the current polytope because the polytope belongs to the $[0,1]^n$ cube. For each subsequent query, $\sum_{i \in I} \alpha_i f_i  \geq t-1$ is valid because the previous query deduced $\sum_{i \in I} \alpha_i f_i  \geq t-1$. Similar arguments show that the remaining queries are also facelike. Thus, \autoref{lem:SPtopSP} completes the proof.
\end{proof}

We note that the $\CP$ refutations that result from  \autoref{thm:linear-equations} have a very particular structure: they are extremely long and narrow. 
Indeed, they have depth $n^{O(\log m)}$. 
We give a rough sketch of the argument: it is enough to show that most lines $L_i$ in the $\CP$ refutation are derived using some previous line $L_j$ with $j=O(i)$. 
This is because the final line would have depth proportional to the size of the proof. 
To see that the $\CP$ refutation satisfies this property, observe that for each node visited in the in-order traversal, the nodes in the right subproof $\pi_R$ depend on the halfspace labelling the root, which in turn depends on the left subproof $\pi_L$.

\section{Lower Bound on the Depth of Semantic $\CP$ Refutations}
Our results from \autoref{sec:SP} suggest an interesting interplay between depth and size of Cutting Planes proofs.
In particular, we note that there is a \emph{trivial} depth $n$ and exponential size refutation of any unsatisfiable CNF formula in Cutting Planes; however, it is easy to see that the Dadush--Tiwari proofs and our own quasipolynomial size $\CP$ proofs of Tseitin are also extremely deep (in particular, they are \emph{superlinear}). Even in  the stronger \emph{Semantic} $\CP$ it is not clear that the depth of these proofs can be  decreased.
However, this does not hold for $\SP$, which has quasi-polynomial size and poly-logarithmic depth refutations.
This motivates \autoref{conj:deepProofs}, regarding the existence of a ``supercritical'' trade-off between size and depth for Cutting Planes \cite{Razborov16,BerkholzN20}. The Tseitin formulas are a natural candidate for resolving this conjecture. 

In this section we develop a new method for proving depth lower bounds which we believe should be more useful for resolving this conjecture. Our method works not only for $\CP$ but also for semantic $\CP$. Using our technique, we establish the first linear lower bounds on the depth of Semantic $\CP$ refutations of the Tseitin formulas. 

Lower bounds on the depth of \emph{syntactic} $\CP$ refutations of Tseitin formulas were established by Buresh-Openheim et al.~\cite{BGHMP06} using a rank-based argument. Our proof is inspired by their work, and so we describe it next. 
Briefly, their proof proceeds by considering a sequence of polytopes $P^{(0)} \supseteq \ldots \supseteq P^{(d)}$ where $P^{(i)}$ is the polytope defined by all inequalities that can be derived in depth $i$ from the axioms in $F$. The goal is to show that $P^{(d)}$ is not empty. To do so, they show that a point $p \in P^{(i)}$ is also in $P^{(i+1)}$ if for every coordinate $j$ such that $0< p_j<1$, there exists points $p^{(j,0)}, p^{(j,1)} \in P^{(i)}$ such that $p^{(j,b)}_k = b$ if $k = j$ and $p^{(j,b)}_k = p_k$ otherwise. The proof of this fact is syntactic: it relies on the careful analysis of the precise rules of $\CP$.


When dealing with Semantic $\CP$, we can no longer analyze a finite set of syntactic rules. Furthermore, it is not difficult to see that the aforementioned criterion for membership in $P^{(i+1)}$ is no longer sufficient for Semantic $\CP$. We develop an analogous criterion for Semantic $\CP$ given later in this section. 
As well, we note that the definition of $P^{(i)}$ is not well-suited to studying the depth of bounded-size  $\CP$ proofs like those in \autoref{conj:deepProofs} --- there does not appear to be a useful way to limit $P^{(i)}$ to be a polytope derived by a bounded number of halfspaces. Therefore we develop our criterion in the language of lifting, which is more amenable to supercritical tradeoffs \cite{Razborov16,BerkholzN20}.

Through this section we will work with the following \emph{top-down} definition of Semantic $\CP$.
\begin{defn}
	Let $F$ be an $n$-variate unsatisfiable CNF formula.
	An $\sCP$ refutation of $F$ is a directed acyclic graph of fan-out $\leq 2$ where each node $v$ is labelled with a halfspace $H_v\subseteq \mathbb{R}^n$ (understood as a set of points satisfying a linear inequality) satisfying the following:
	\begin{enumerate}
		\item \emph{Root.} There is a unique source node $r$ labelled with the halfspace $H_v=\mathbb{R}^n$ (corresponding to the trivially true inequality $1 \geq 0$).
		\item \emph{Internal-Nodes.} For each non-leaf node $u$ with children $v, w$, we have \[H_u \cap \boolcube^n \subseteq H_v \cup H_{w}.\]
		\item \emph{Leaves.} Each sink node $u$ is labeled with a unique clause $C \in F$ such that $H_v \cap \boolcube^n \subseteq C^{-1}(0)$.
	\end{enumerate}
\end{defn}
The above definition is obtained by taking a (standard) $\sCP$ proof and \emph{reversing all inequalities}: now, a line is associated with the set of assignments \emph{falsified} at that line, instead of the assignments \emph{satisfying} the line.

To prove the lower bound we will need to find a long path in the proof.
To find this path we will be taking a root-to-leaf walk down the proof while constructing a partial restriction $\rho \in \set{0,1,*}^n$ on the variables. 
For a partial restriction $\rho$, denote by $\free(\rho) := \rho^{-1}(*)$ and $\fix(\rho):= [n] \setminus \free(\rho)$. 
Let the \emph{restriction} of $H$ by $\rho$ be the halfspace
\[H \restriction \rho := \{ x \in \reals^{\free(\rho)}: \exists \alpha \in H,~ \alpha_{\fix(\rho)} = \rho_{\fix(\rho)},~ \alpha_{\free(\rho)} = x \}. \]
It is important to note that $H \restriction \rho$ is itself a halfspace on the \emph{free} coordinates of $\rho$.

One of our key invariants needed in the proof is the following.
\begin{defn}\label{def:good-halfspace}
	A halfspace $H \subseteq \reals^n$ is \emph{good} if it contains the all-$\frac{1}{2}$ vector, that is, $(\frac{1}{2})^n=(\frac{1}{2},\frac{1}{2},\ldots,\frac{1}{2}) \in H$.
\end{defn}

We will need two technical lemmas to prove the lower bounds.
The first lemma shows that if a good halfspace $H$ has its boolean points covered by halfspaces $H_1, H_2$, then one of the two covering halfspaces is also good modulo restricting a small set of coordinates.
\begin{lem}\label{lem:crux}
	Let $H \subseteq \reals^n$ be any good halfspace, and suppose $H \cap \boolcube^n \subseteq H_1 \cup H_2$ for halfspaces $H_1, H_2$.
	Then there is a restriction $\rho$ and an $i = 1,2$ such that $|\fix(\rho)| \leq 2$ and $H_i \restriction \rho$ is good.
\end{lem}

The second lemma shows that good halfspaces are \emph{robust}, in the sense that we can restrict a good halfspace to another good halfspace while also satisfying any mod-2 equation.
\begin{lem}\label{lem:consistent-restriction}
	Let $n \geq 2$ and $H \subseteq \reals^n$ be a good halfspace.
	For any $I \subseteq [n]$ with $|I| \geq 2$ and $b \in \{0,1\}$, there is a partial restriction $\rho \in \set{0,1,*}^n$ with $\fix(\rho) = I$ such that 
\begin{itemize}
	\item $\displaystyle{\bigoplus_{i \in I} \rho(x_i) = b}$ and
	\item $H \restriction \rho \subseteq \reals^{\free(\rho)}$ is good.
\end{itemize}
\end{lem}

With these two lemmas one can already get an idea of how to construct a long path in the proof.
Suppose we start at the root of the proof; the halfspace is $1 \geq 0$ (which is clearly good) and the restriction we maintain is $\rho = *^n$.
We can use the first lemma to move from the current good halfspace to a good child halfspace while increasing the number of fixed coordinates by at most $2$.
However, we have no control over the two coordinates which are fixed by this move, and so we may fall in danger of falsifying an initial constraint.
Roughly speaking, we will use the second lemma to satisfy constraints that are in danger of being falsified.

We delay the proofs of these technical lemmas to the end of the section, and first see how to prove the depth lower bounds.

\subsection{Lifting Decision Tree Depth to Semantic $\CP$ Depth}

As a warm-up, we show how to lift lower bounds on Resolution depth to Semantic $\CP$ depth by composing with a constant-width XOR gadget.
If $F$ is a CNF formula then we can create a new formula by replacing each variable $z_i$ with an XOR of $4$ new variables $x_{i,1},\ldots, x_{i,4}$:
\[z_i := \XOR_4(x_{i,1},\ldots, x_{i,4}) = x_{i,1} \oplus \cdots \oplus x_{i,4}. \] 
We call $z_i$ the \emph{unlifted} variable associated with the output of the $\XOR_4$ \emph{gadget} applied to the $i$-th \emph{block} of variables. 
Formally, let $\XOR_4^n:\{0,1\}^{4n} \rightarrow \{0,1\}^n$ be the application of $\XOR_4$ to each $4$-bit block of a $4n$-bit string. 
Let $F \circ \XOR_4^n$ denote the formula obtained by performing this substitution on $F$ and transforming the result into a CNF formula in the obvious way.

The main result of this section is the following.
\begin{thm}\label{thm:cp-depth-lifting}
	For any unsatisfiable CNF formula $F$, \[\depth_{\sCP}(F\circ \XOR_4^n) \geq \frac{1}{2} \depth_{\Res}(F).\]
\end{thm}

Key to our lower bound will be the following characterization of Resolution depth by \emph{Prover-Adversary} games.

\begin{defn}\label{def:prover-adversary}
	The \emph{Prover--Adversary} game associated with an $n$-variate formula $F$ is played between two competing players, Prover and Adversary. 
	The game proceeds in rounds, where in each round the state of the game is recorded by a partial assignment $\rho\in\{0,1,*\}^n$ to the variables of $F$. 

	Initially the state is the empty assignment $\rho=*^n$. 
	Then, in each round, the Prover chooses an $i \in [n]$ with $\rho_i = *$, and the Adversary chooses $b \in \boolcube$.
	The state is updated by $\rho_i \leftarrow b$ and play continues.
	The game ends when the state $\rho$ falsifies an axiom of $F$. 

	It is known \cite{Pudlak00} that $\depth_{\Res}(F)$ is exactly the smallest $d$ for which there is a Prover strategy that ends the game in $d$ rounds, regardless of the strategy for the Adversary.
\end{defn}

The proof of \autoref{thm:cp-depth-lifting} will follow by using an optimal Adversary strategy for $F$ to construct a long path in the Semantic $\CP$ proof of $F \circ \XOR_4^n$.
Crucially, we need to understand how halfspaces $H$ transform under $\XOR_4^n$:
\[ \XOR_4^n(H):= \{z \in \{0,1\}^n : \exists x \in H \cap \set{0,1}^{4n}, \XOR^n_4(x) = z\}. \]
As we have already stated, we will maintain a partial assignment $\rho \in \{0,1,*\}^{4n}$ on the $4n$ \emph{lifted} variables.
However, in order to use the Adversary, we will need to convert $\rho$ to a partial assignment on the $n$ \emph{unlifted} variables.
To perform this conversion, for any $\rho \in \set{0,1,*}^{4n}$ define $\XOR_4^n(\rho) \in \{0,1,*\}^n$ as follows: for each block $i \in [n]$, define \[ \XOR^n_4(\rho)_i  = \begin{cases}
		\XOR_4(\rho(x_{i,1}),\ldots, \rho(x_{i, 4})) &\mbox{ if $(i,j) \in \fix(\rho)$ for $j \in [4]$}, \\
 		* &\mbox{ otherwise}.
 \end{cases}
 \]

We are now ready to prove \autoref{thm:cp-depth-lifting}.
Fix any Semantic $\CP$ refutation of $F \circ \XOR_4^n$, and suppose that there is a strategy for the Adversary in the Prover-Adversary game of $F$ certifying that $F$ requires depth $d$.
Throughout the walk, we maintain a partial restriction $\rho \in \set{0,1,*}^{4n}$ to the lifted variables satisfying the following three invariants with respect to the current visited halfspace $H$. 
\begin{enumerate}
	\item[--] \emph{Block Closed}. In every block either all variables in the block are fixed or all variables in the block are free.
	\item[--] \emph{Good Halfspace}. $H \restriction \rho$ is good.
	\item[--] \emph{Strategy Consistent}. The unlifted assignment $\XOR_4^n(\rho)$ does not falsify any clause in $F$.
\end{enumerate}
Initially, we set $\rho = *^{4n}$ and the initial halfspace is $1 \geq 0$, so the pair $(H, \rho)$ trivially satisfy the invariants.
Suppose we have reached the halfspace $H$ in our walk and $\rho$ is a restriction satisfying the invariants.
We claim that $H$ cannot be a leaf.
To see this, suppose that $H$ is a leaf, then by definition $H \cap \boolcube^{4n} \subseteq C^{-1}(0)$ for some clause $C \in F \circ \XOR^n_4$. 
By the definition of the lifted formula, this implies that $\XOR^n_4(H) \subseteq D^{-1}(0)$ for some clause $D \in F$.
Since $(H, \rho)$ satisfy the invariants, the lifted assignment $\XOR_4^n(\rho)$ does not falsify $D$, and so by the block-closed property it follows that there must be a variable $z_i \in D$ such that all lifted variables in the block $i$ are free under $\rho$.
But then applying \autoref{lem:consistent-restriction} to the block of variables $\set{x_{i,1}, x_{i,2}, x_{i,3}, x_{i,4}}$, we can extend $\rho$ to a partial assignment $\rho'$ such that $z_i = \XOR_4(\rho(x_{i,1}), \rho(x_{i,2}), \rho(x_{i,3}), \rho(x_{i,4}))$ satisfies $D$.
But $H \restriction \rho'$ is a projection of $H \restriction \rho$ and so this contradicts that $\XOR^n_4(H)$ violates $D$.

It remains to show how to take a step down the proof. 
Suppose that we have taken $t < d/2$ steps down the Semantic $\CP$ proof, the current node is labelled with a halfspace $H$, and the partial assignment $\rho$ satisfies the invariants.
If $H$ has only a single child $H_1$, then $H \cap \set{0,1}^{4n} \subseteq H_1 \cap \boolcube^{4n}$ and $\rho$ will still satisfy the invariants for $H_1$.
Otherwise, if $H$ has two children $H_1$ and $H_2$ then applying \autoref{lem:crux} to the halfspaces $H \restriction \rho, H_1 \restriction \rho, H_2 \restriction \rho$ we can find an $i \in \{1,2\}$ and a restriction $\tau$  such that $H_i \restriction (\rho\tau)$ is good and $\tau$ restricts at most $2$ extra coordinates.
Let $i_1, i_2 \in [n]$ be the two blocks of variables in which $\tau$ restricts variables, and note that it could be that $i_1 = i_2$.
	
Finally, we must restore our invariants.
We do this in the following three step process.
\begin{itemize}
	\item Query the Adversary strategy at the state $\XOR_4^n(\rho)$ on variables $z_{i_1}, z_{i_2}$ and let $b_1, b_2 \in \{0,1\}$ be the responses.
	\item For $i = i_1, i_2$ let $I_i$ be the set of variables free in the block $i$, and note that $|I_i| \geq 2$. 
		Apply \autoref{lem:consistent-restriction} to $H \restriction (\rho\tau)$ and $I_i$ to get new restrictions $\rho_{i_1}, \rho_{i_2}$ so that blocks $i_1$ and $i_2$ both take values consistent with the Adversary responses $b_1, b_2$.
	\item Update $\rho \leftarrow \rho\tau\rho_{i_1}\rho_{i_2}$.
\end{itemize} 
By \autoref{lem:consistent-restriction} the new restriction $\rho$ satisfies the block-closed and the good halfspace invariants.
At each step we fix at most two blocks of variables, and thus the final invariant is satisfied as long as $t < d/2$.
This completes the proof.

\subsection{Semantic $\CP$ Depth Lower Bounds for Unlifted Formulas}
Next we show how to prove depth lower bounds directly on \emph{unlifted} families of $\mathbb{F}_2$-linear equations. The strength of these lower bounds will depend directly on the expansion of the underlying constraint-variable graph of $F$. 

Throughout this section, let $F$ denote a set of $\mathbb{F}_2$-linear equations.
In a Semantic $\CP$ proof, we must encode $F$ as a CNF formula, but while proving the lower bound we will instead work with the underlying system of equations.
For a set $F$ of $\mathbb{F}_2$-linear equations let $G_F := (F \cup V, E)$ be the bipartite \emph{constraint-variable} graph defined as follows. 
Each vertex in $F$ corresponds to an equation in $F$ and each vertex in $V$ correspond to variables $x_i$. 
There is an edge $(C_i,x_j) \in E$ if $x_j$ occurs in the equation $C_i$.
For a subset of vertices $X \subseteq F \cup V$ define the \emph{neighbourhood} of $X$ in $G_F$ as $\Gamma(X):= \{v \in F \cup V : \exists u \in X, (u,v) \in E\}$.

\begin{defn}\label{def:expansion}
	For a bipartite graph $G=(U \cup V, E)$ the \emph{boundary} of a set $W \subseteq U$ is \[\delta(W):= \{v \in V : |\Gamma(v) \cap W| = 1\}.\]
	The \emph{boundary expansion} of a set $W \subseteq U$ is $|\delta(W)|/|W|$. 
	The graph $G$ is a $(r, s)$-\emph{boundary expander} if the boundary expansion of every set $W \subseteq U$ with $|W| \leq r$ has boundary expansion at least $s$.
\end{defn}

If $F$ is a system of linear equations then we say that $F$ is an $(r,s)$-boundary expander if its constraint graph $G_{F}$ is.
The main result of this section is the following theorem, analogous to \autoref{thm:cp-depth-lifting}.

\begin{thm}\label{thm:expander-lb}
	For any system of $\mathbb{F}_2$-linear equations $F$ that is an $(r, s+3)$-boundary expander, \[\depth_{\sCP}(F) \geq rs/2.\]
\end{thm}

The proof of this theorem follows the proof of \autoref{thm:cp-depth-lifting} with some small changes.
As before, we will maintain a partial assignment $\rho \in \{0,1,*\}^n$ that will guide us on a root-to-leaf walk through a given Semantic $\CP$ proof; we also require that each halfspace $H$ that we visit is \emph{good} relative to our restriction $\rho$.
Now our invariants are (somewhat) simpler: we will only require that $F \restriction \rho$ is a sufficiently good boundary expander.

We first prove an auxiliary lemma that will play the role of \autoref{lem:consistent-restriction} in the proof of \autoref{thm:expander-lb}.
We note that it follows immediately from \autoref{lem:consistent-restriction} and boundary expansion.
\begin{lem}\label{lem:boundary-cleanup}
	Suppose $F$ is a system of $\mathbb{F}_2$-linear equations that is an $(r, s)$-boundary expander for $s > 1$, and suppose $F' \subseteq F$ with $|F'| \leq r$.
	Let $H$ be a good halfspace.
	Then there exists a $\rho \in \set{0,1,*}^n$ with $\fix(\rho) = \Gamma(F')$ such that 
	\begin{itemize}
		\item $F'$ is satisfied by $\rho$, and
		\item $H \restriction \rho$ is good.
	\end{itemize}
\end{lem}
\begin{proof}
	We first use expansion to find, for each constraint $C_i \in F'$, a pair of variables $y_{i,1}, y_{i,2}$ that are in $C_i$'s boundary. To do this, first observe that $|\delta(F')| \geq s |F'| > |F'|$ by the definition of boundary expansion.
	The pigeonhole principle then immediately implies that there are variables $y_{i,1}, y_{i,2} \in \delta(F')$ and a constraint $C_i \in F'$ such that $y_{i,1}, y_{i,2} \in C_i$. 
	Since $y_{i,1},y_{i,2}$ do not occur in $F' \setminus \{C_i\}$, it follows that $F' \setminus \{C_i\}$ is still an $(r, s)$-boundary expander. 
	So, we update $F' = F' \setminus \set{C_i}$ and repeat the above process.

	When the process terminates, we have for each constraint $C_i \in F'$ a pair of variables $y_{i,1}, y_{i,2}$ that occur \emph{only} in $C_i$.
	Write the halfspace $H = \sum_{i} w_i x_i \geq c$, and let $I = \Gamma(F') \setminus \bigcup_{i \in I} \set{y_{i,1}, y_{i,2}}$ be the set of variables occurring in $F'$ that were not collected by the above process.
	We define a partial restriction $\rho$ with $\fix(\rho) = I$ that depends on $|I|$ as follows.	
	\begin{itemize}
		\item If $|I| = 0$ then $\rho = *^n$.
		\item If $I = \set{x_i}$ then define $\rho(x_i) = 1$ if $w_i \geq 0$ and $\rho(x_i) = 0$ otherwise, and for all other variables set $\rho(x) = *$.
		\item If $|I| > 2$ then apply \autoref{lem:consistent-restriction} to generate a partial restriction $\rho$ with $\fix(\rho) = I$ that sets the $\XOR$ of $I$ arbitrarily.
	\end{itemize}
	Observe that $H \restriction \rho$ is good.
	The only non-trivial case is when $|I| = 1$, but, in this case we observe \[ (H \restriction \rho)((1/2)^{n-1}) = w_i \rho(x_i) + \sum_{j \neq i} w_i/2 \geq \sum_{i} w_i/2 \geq c,\]
	where we have used that $H$ is good and the definition of $\rho$.

	Next we extend $\rho$ as follows: for each $i = 1, 2, \ldots, |F'|$ apply \autoref{lem:consistent-restriction} to $I_i = \set{y_{i,1}, y_{i,2}}$ to generate a partial restriction $\rho_i$ with $\fix(\rho_i) = I_i$ so that the constraint $C_i \restriction \rho \rho_1 \cdots \rho_{i-1}$ is satisfied by $\rho_i$.
	Observe that this is always possible since $I_i$ is in the boundary of $C_i$.
	Finally, we update $\rho \leftarrow \rho \rho_1 \cdots \rho_{|F'|}$.
	It follows by \autoref{lem:consistent-restriction} that $F'$ is satisfied by $\rho$ and $H \restriction \rho$ is good.
\end{proof}

	We are now ready to prove \autoref{thm:expander-lb}.
	Fix any Semantic $\CP$ refutation of $F$ and let $n$ be the number of variables.
	We take a root-to-leaf walk through the refutation while maintaining a partial assignment $\rho \in \set{0,1,*}^n$ and an integer valued parameter $k \geq 0$.
	Throughout the walk we maintain the following invariants with respect to the current halfspace $H$: 
	\begin{itemize}
		\item[--] \emph{Good Expansion.} $F \restriction \rho$ is a $(k, t)$-boundary expander with $t > 3$.
		\item[--] \emph{Good Halfspace.} $H \restriction \rho$ is good.
		\item[--] \emph{Consistency.} The partial assignment $\rho$ does not falsify any clause of $F$.
	\end{itemize}
	
	Initially, we set $k = r$, $\rho = *^n$, and $t=s+3$, so the invariants are clearly satisfied since $F$ is an $(r, s+3)$-expander.
	So, suppose that we have reached a halfspace $H$ in our walk, and let $k, \rho$ be parameters satisfying the invariants.
	We first observe that if $k > 0$ then $H$ cannot be a sink node of the proof.
	To see this, it is enough to show that $H$ contains a satisfying assignment for each equation $C \in F$. 
	Because $H \restriction \rho$ is non-empty (since it is good) there exists a satisfying assignment in $H$ for every equation satisfied by $\rho$, so, 	assume that $C$ is not satisfied by $\rho$. 
	In this case, since $F \restriction \rho$ is a $(k, t)$-expander for $k>0$ we can apply \autoref{lem:boundary-cleanup} to $\set{C}$ and $H \restriction \rho$ and obtain a partial restriction $\tau$ with $\fix(\tau) = \Gamma(C)$ such that $\tau$ satisfies $C$.
	It follows that $H$ is not a leaf.

	Next, we show how to take a step down the proof while maintaining the invariants. 
	If $H$ has only a single child $H_1$, then $H \subseteq H_1$ and we can move to $H_1$ without changing $\rho$ or $k$.
	Otherwise, let the children of $H$ be $H_1$ and $H_2$. 
	Applying \autoref{lem:crux} to $H \restriction \rho, H_1 \restriction \rho, H_2 \restriction \rho$ we get a partial restriction $\tau$ and an $i \in \set{1,2}$ such that $H_i \restriction \rho\tau$ is good and $|\fix(\tau)| \leq 2$.
	Due to this latter fact, since $F \restriction \rho$ is a $(k, t)$-expander it follows that $F \restriction \rho\tau$ is a $(k, t-2)$-expander in the worst case.
	Observe that since $t > 3$ it follows that $F \restriction \rho\tau$ still satisfies the consistency invariant.
	It remains to restore the expansion invariant.

	To restore the expansion invariant, let $W$ be the largest subset of equations such that $|W| \leq k$ and $W$ has boundary expansion at most $3$ in $F \restriction \rho\tau$, and note that $W$ has boundary expansion at least $t-2 > 1$.
	Applying \autoref{lem:boundary-cleanup}, we can find a restriction $\rho'$ such that $W \restriction \rho\tau\rho'$ is satisfied, and $H \restriction \rho\tau\rho'$ is a good halfspace.
	Since $W$ is the largest subset with expansion at most $3$, it follows that $F \restriction \rho \tau \rho'$ is now a $(k-|W|, t')$-boundary expander with $t' > 3$. Suppose otherwise, then there exists a subset of equations $W'$ which has boundary expansion at most $3$ in $F \restriction \rho \tau \rho'$. Then $W \cup W'$ would have had boundary expansion at most $3$ in $F \restriction \rho \tau$, contradicting the maximality of $W$. Now update $\rho \leftarrow \rho\tau\rho'$ and $k \leftarrow k - |W|$.
	 Finally, we halt the walk if $k = 0$.

	We now argue that this path must have had depth at least $rs /2$ upon halting.
	Assume that we have taken $t$ steps down the proof. 
	For each step $i \leq t$ let $W_i$ be the set of equations which lost boundary expansion during the $i$th cleanup step. 
	Note that $W_i \cap W_j = \emptyset$ for every $i \neq j$. 
	Let $W^* = \cup_{i=1}^t W_i$, note that $|W^*| = r$ because at the $i$th step we decrease $k$ by $|W_i|$. Furthermore, at the end of the walk, $W^*$ has no neighbours and therefore no boundary in $F\restriction \rho$. 
	Before the start of the $i$th cleanup step, $W_i$ has at most $3|W_i|$ boundary variables. 
	Therefore, at most $3|W^*| = 3r$ boundary variables were removed during the cleanup step. 
	Since $F$ started as an $(r, s+3)$-boundary expander, it follows that $W^*$ had at least $r(s+3)$ boundary variables at the start of the walk.
	But, since \emph{all} variables have been removed from the boundary by the end, this means that $rs$ variables must have been removed from the boundary during the move step.
	Thus, as each move step sets at most $2$ variables, it follows that $t \geq rs/2$ before the process halted.

\subsection{Proof of \autoref{lem:crux} and \autoref{lem:consistent-restriction}}
In this section we prove our two key technical lemmas: \autoref{lem:crux} and \autoref{lem:consistent-restriction}.
We begin by proving \autoref{lem:consistent-restriction} as it is simpler.

\begin{proof}[Proof of \autoref{lem:consistent-restriction}]
	Let $H$ be represented by $\sum_{i \in [n]} w_i x_i \geq c$ and suppose without loss of generality that $c \geq 0$ and that $I = \{1, \ldots, k\}$. 
	Let the weights of $I$ in $H$ be ordered $|w_1| \geq |w_2| \geq \ldots |w_k|$. 
	Define $\rho$ by setting $\rho(x_i) = *$ for $i \not \in I$, for $i \leq k-1$ set $\rho(x_i) = 1$ if $w_i \geq 0$ and $\rho(x_i) = 0$ otherwise, and set $\rho(x_k)$ so that $\bigoplus_{i \in I} \rho(x_i) = b$. 
	Clearly the parity constraint is satisfied, we show that $H \restriction \rho$ is good.
	This follows by an easy calculation:
	\begin{align*} 
		(H \restriction \rho)((1/2)^{[n]\setminus I}) & = w_{k-1}\rho(x_{k-1}) + w_k\rho(x_k) + \sum_{i \leq k-2} w_i \rho(x_i) + \sum_{i \geq k+1} w_i/2 \\
			&\geq w_{k-1}/2 + w_k/2 + \sum_{i \leq k-2} w_i \rho(x_i) + \sum_{i \geq k+1} w_i/2 	\\
			&\geq \sum_{i \in [n]} w_i/2 \geq c 
	\end{align*}
	where the first inequality follows by averaging since $|w_{k-1}| \geq |w_{k}|$, and the final inequality follows since $H$ is good. 
\end{proof}

In the remainder of the section we prove \autoref{lem:crux}.
It will be convenient to work over $\{-1,1\}^n$ rather than $\{0,1\}^n$, so, we restate it over this set and note that we can move between these basis by using the bijection $v \mapsto (1-v)/2$.

\begin{lem}\label{lem:CruxForPM1}
	Let $H \in \reals^n$ be a halfspace such that $0^n \in H$ and suppose that $H \cap \set{-1,1}^n \subseteq H_1 \cup H_2$. 
	Then one of $H_1$ or $H_2$ contains a point $y \in \{-1,0,1\}^n$ such that $y$ has at most two coordinates in $\{-1,1\}$. 
\end{lem}

The key ingredient in our proof of \autoref{lem:CruxForPM1} is the following simple topological lemma, which will allow us to find a well-behaved point lying on a $2$-face of the $\{-1,1\}^n$ cube

\begin{defn}[$2$-face]
	A $2$-face of the $n$-cube with vertices $\{-1,1\}^n$ are the $2$-dimensional $2$-by-$2$ squares spanned by four vertices of the cube that agree on all but two coordinates. 
That is, a two face is a set $A \subseteq [-1,1]^n$ such that there exists $\rho \in \{-1,1,*\}^n$ with $|\free(\rho)| = 2$ and $A = [-1,1]^n \restriction \rho$.
\end{defn}

\begin{lem} \label{lem:orthogonal_vector}
Let $w_1, w_2 \in \mathbb{R}^n$ be any pair of non-zero vectors, then we can find a vector $v \in \mathbb{R}^n$ orthogonal to $w_1, w_2$, such that $v$ lies on a $2$-face.
\end{lem}
\begin{proof}
	We will construct the vector $v$ iteratively by rounding one coordinate at a time to a $\{-1,1\}$-value until $v$ contains exactly $n-2$ coordinates fixed to $\{-1,1\}$. At each step, we will maintain that $v \in [-1,1]^n$ and that $v$ is orthogonal to $w_1$ and $w_2$. Therefore when the process halts $v$ will lie on a $2$-face.
	
	Initially, set $v = 0^n$ and observe that the invariants are satisfied. Suppose that we have constructed a vector $v$ that is orthogonal to $w_1$ and $w_2$, all of its coordinates belong to $[-1,1]$, and exactly $i < n-2$ of its coordinates belong to $\{-1,1\}$; suppose w.l.o.g. that they are the first $i$ coordinates. We will show how to ``booleanize'' an additional coordinate of $v$. Let $u$ be any non-zero vector that is orthogonal to $\{w_1, w_2, e_1, \ldots, e_i\}$, where $e_j$ is the $j$th standard basis vector. Begin moving from $v$ in the direction of $u$ and let $\alpha > 0$ be the smallest value such that one of the coordinates $j > i$ of $v + \alpha u$ is in $\{-1, 1\}$. 
	We verify that the following properties hold:
	\begin{enumerate}
		\item The first $i$ coordinates of $v + \alpha u$ are in $\{ -1, 1\}$. This follows because we moved in a direction that is orthogonal to $e_1, \ldots, e_i$. 
		\item $v + \alpha u$ is orthogonal to $w_1$ and $w_2$. Let $w$ be either of the vectors $w_1$ or $w_2$ and observe that  $v_{i+1}w = v_iw + \alpha (uw) = 0$,
		where the final equality follows because $w$ is orthogonal to $v_i$ by induction and to $u$ by assumption. 
	\end{enumerate}
	Finally, set $v$ to be  $v + \alpha u$. 
\end{proof}

\begin{proof}[Proof of \autoref{lem:CruxForPM1}]
	Let the children $H_1$ and $H_2$ of $H$ be given by the halfspaces $w_1x \geq b_1$ and $w_2x \geq b_2$ respectively. By \autoref{lem:orthogonal_vector} we can find a vector $v$ which is orthogonal to $w_1$ and $w_2$, and which lies on some $2$-face $F$ of the $[-1,1]^n$ cube corresponding to some restriction $\rho \in \{0,1,*\}^n$. Then, $v$ lies in (at least) one of the four 1-by-1 quadrants of the $2$-face, $[0,1]^2$, $[0,1] \times [-1,0]$, $[-1,0] \times [0,1]$, or $[-1,0]^2$; suppose that $v$ lies in the $[-1,0] \times [0,1]$ quadrant of $F$, the other cases will follow by symmetry (see \autoref{fig:2Face}).

		\begin{figure}[H]
		\centering
		\begin{tikzpicture}
		\draw[very thick, color=white, fill=green!9] (0,0)  -- (0,2) -- (-2,2) -- (-2,0) -- cycle;
		\draw[very thick] (-2,2)  -- (2,2) -- (2,-2) -- (-2,-2) -- cycle;
		\draw[thick, color = gray] (0,-2) -- (0,2);
		\draw[thick, color = gray] (-2,0) -- (2,0);
		\draw[very thick, ->, color=red!50] (0,0)  -- (-2,.4) node[midway, above] {$v$};
		\draw[very thick, ->, color=red!50] (-2,.4)  -- (-2,2) node[midway, left] {$a-v$};

		\node[text width=.5cm] at (0.3,-0.3) {$(0,0)$};
		\node[text width=2cm] at (2.6 ,2.4) {$(1,1)$};
		\node[text width=2.2cm] at (2.5 ,-2.4) {$(1,-1)$};
		\node[text width=2.2cm] at (-2.2 ,2.4) {$a=(-1,1)$};
		\node[text width=2.4cm] at (-1.8 ,-2.4) {$(-1,-1)$};
		
		 \end{tikzpicture}
		 \caption{A $2$-face of the $n$-cube together with a depiction of the booleanizing process.} \label{fig:2Face}	
	\end{figure}
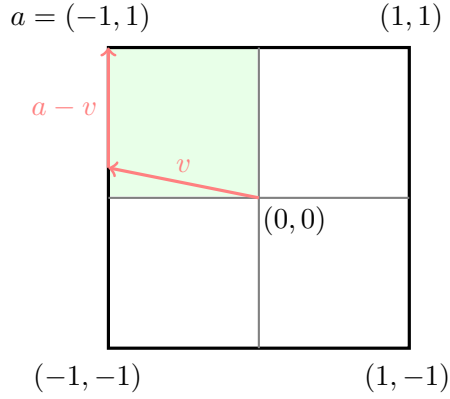
	
	Let $a \in \mathbb{R}^n$ be the vector corresponding to the $(-1,1)$ corner of $F$, i.e., $a$ is $\rho$ extended by setting the two free bits to $-1$ and $1$. By symmetry and the fact that $H$ is good (and therefore $0^n \in H$), we can assume that $a$ is contained in $H$ --- otherwise, simply exchange $a$ and $v$ for $-a$ and $-v$.
	Since $H \cap \set{-1, 1}^n \subseteq H_1 \cup H_2$ and $a \in \set{-1,1}^n$, it follows that $a$ is in one of $H_1$ or $H_2$.  
	Assume that $a \in H_1$; that is, $w_1a \geq b_1$.
	Our goal is to construct a vector $y \in H_1$ that satisfies the statement of the lemma. Consider the following two cases:
	\begin{enumerate}
		\item[(i)] If $w_1( a-v) \leq 0$, then it follows that  $y := 0^n \in H_1$. Indeed,
	$w_1 y = w_1v \geq w_1 a \geq b_1,$
	where first equality follows because $w_1$ and $p$ are orthogonal by assumption, and the final inequality follows because $a \in H_1$. 
		\item[(ii)] Otherwise, we have that $w_1(a-v) > 0$. 
			We construct a point that satisfies the statement of the lemma as follows.
			First, note that since $a, v \in F$, it follows that the vector $a - v$ has at most two non-zero coordinates.
			Beginning at the origin $0^n$, move in the direction $a-v$ until a free coordinate coordinate becomes fixed to $-1$ or $1$; that is, let $\alpha>0$ be the minimum value such that $\alpha(a-v)$ has at most one coordinate which is not $\{-1,1\}$-valued.
			Since both $a$ and $v$ belong to the same $1 \times 1$ quadrant of the $2$-face, $\|a-v\|_{\infty} \leq 1$ and so $\alpha \geq 1$.
			We can then verify that $\alpha(a-v) \in H_1$, since
			\[w_1 \alpha(a-v) = \alpha(w_1a) - 0 \geq w_1a \geq b_1,\]
			where we have used the fact that $v$ is orthogonal to $w_1$ and $\alpha \geq 1$.	
			Finally, since $\alpha(a - v) \in H_1$ we can round the final non-zero coordinate to $-1$ or $1$; since $H_1$ is a halfspace one of the two vectors will remain in $H_1$.
	\end{enumerate}
	
\end{proof}

\subsection{Applications}
We now use the theorems from the previous sections to obtain several concrete lower bounds.
First, we give strong depth lower bounds for $\sCP$ proofs of Tseitin formulas on expander graphs. 
\begin{thm}\label{cor:TseitinLowerBound}
	There exists a graph $G$ and labelling $\ell:V \rightarrow \{0,1\}$ such that any $\sCP$ refutation of $\Tseitin(G,\ell)$ requires depth $\Omega(n)$. 
\end{thm}
\begin{proof}
	A graph $G=(V,E)$ is a $\gamma$-\emph{vertex expander} if 
	\[\min \set{|\Gamma(W)| : W \subseteq V, |W| \leq |V|/2} \geq \gamma|W|,\]
	where $\Gamma(W)$ is the neighbourhood of $W$.
	We claim that if $G$ is a $\gamma$-vertex expander then any Tseitin formula over $G$ is a $(n/2, \gamma)$-boundary expander.
	Fix any subset $W$ of the equations with $|W| \leq n/2$. 
	By the definition of vertex expansion we have that $|\Gamma(W)| \geq \gamma|W|$, and since each variable is contained in exactly two constraints, it follows that the boundary of $W$ in $\Tseitin(G,\ell)$ has size at least $|\delta(W)| \geq \gamma|W|$.
	The result then follows from \autoref{thm:expander-lb} and the existence of strong vertex expanders $G$ (e.g.~$d$-regular Ramanujan graphs are at least $d/4$-vertex expanders, and exist for all $d$ and $n$ \cite{MarcusSS18}).
\end{proof}

Next, we give lower bounds on the depth of Semantic $\CP$ refutations of random $k$-XOR and random $k$-CNF formulas for constant $k$.
\begin{defn}
	Let $\XOR(m,n,k)$ be the distribution on random $k$-$\XOR$ formulas obtained by sampling $m$ equations from the set of all $\bmod\ 2$ linear equations with exactly $k$ variables.
\end{defn}
\begin{thm}\label{thm:randomLB}
The following holds for Semantic $\CP:$
	\begin{enumerate}
		\item[$1$.] For any $k \geq 6$ there exists $m = O(n)$ such that $F \sim \XOR(m,n,k)$ requires refutations of depth at least $\Omega(n)$ with high probability.
		\item[$2$.] For any $k \geq 6$ there exists $m = O(n)$ such that $F \sim \mathcal{F}(m,n,k)$ requires refutations of depth at least $\Omega(n)$ with high probability.
	\end{enumerate}
\end{thm}
\begin{proof}
	We first prove (1) and obtain (2) via a reduction. 
	Fix $m= O(n)$ so that $F$ is unsatisfiable with high probability. 
	For any constant $k,\delta$ and $m=O(n)$,  $F \sim \XOR(m,n,k)$ is an $(\alpha n, k-2-2\delta)$-boundary expander for some $\alpha > 0$ (see e.g. \cite{BGHMP06,cs-resolution}). 
	Thus, setting $k \geq 6$ and $\varepsilon$ to be some small constant, the boundary expansion of $G_F$ is at least $3$. 
	By \autoref{thm:expander-lb}, $F$ requires depth $\Omega(n)$ to refute in Semantic $\CP$ with high probability.
	
	The proof of (2) is via a reduction from $\mathcal{F}(m,n,k)$ to $\XOR(m,n,k)$. 
	Every $k$-clause occurs in the clausal encoding of exactly one $k$-$\XOR$ constraint. 
	It follows that from any $k$-CNF formula $F$ we can generate a $k$-$\XOR$ formula whose clausal expansion $F'$ contains $F$ as follows:
	for each clause $C \in F$, if $C$ contains an even (odd) number of positive literals then add to $F'$ every clause on the variables of $C$ which contains an even (odd) number of positive literals. The resulting $F'$ is the clausal encoding of a set of $|F|$ $k$-$\XOR$ constraints. As there is a unique $k$-$\XOR$ consistent with the clauses of $F$, we can define the distribution $\XOR(m,n,k)$ equivalently as follows: 
	\begin{enumerate}
		\item Sample $F \sim \mathcal{F}(m,n,k)$,
		\item Return the $k$-XOR $F'$ generated from $F$ according to the aforementioned process.
	\end{enumerate}
	It follows that the complexity of refuting $F \sim \mathcal{F}(m,n,k)$ is at least that of refuting $F' \sim \XOR(m,n,k)$ and (2) follows from (1) with the same parameters. 
\end{proof}

Finally, we use \autoref{thm:expander-lb} to extend the integrality gaps from \cite{BGHMP06} to $\sCP$ by essentially the same argument. 
For a linear program with constraints given by a system of linear inequalities $Ax \leq b$, the \emph{$r$-round $\sCP$ relaxation} adds all inequalities that can be derived from $Ax \leq b$ by a depth-$r$ $\sCP$ proof.
We show that the $r$-round Semantic $\sCP$ linear program relaxation cannot well-approximate the number of satisfying assignments to a random $k$-SAT or $k$-XOR instance. 

First we define our LP relaxations.
Suppose that $F$ is a $k$-CNF formula with $m$ clauses $C_1, C_2, \ldots, C_m$ and $n$ variables $x_1, x_2, \ldots, x_n$.
If $C_i = \bigvee_{i \in P} x_i \vee \bigvee_{i \in N} \overline x_i$ then let $E(C_i) = \sum_{i \in P} x_i + \sum_{i \in N} 1-x_i$.
We consider the following LP relaxation of $F$:
\begin{align*}
	& \max \sum_{i=1}^m y_i \\
	\text{subject to}\quad & E(C_i) \geq y_i \quad \forall i \in [m]\\
	& 0 \leq x_j \leq 1 \quad \forall j \in [n] \\
	& 0 \leq y_i \leq 1 \quad \forall i \in [m]
\end{align*}

If $F$ is a $k$-XOR formula with $m$ constraints and $n$ variables then we consider the above LP relaxation obtained by writing $F$ as a $k$-CNF.
Finally, recall that the \emph{integrality gap} is the ratio between the optimal integral solution to a linear program and the optimal solution produced by the LP. 

\begin{thm}
	For any $\varepsilon >0$ and $k \geq 6$,
	\begin{enumerate}
		\item[$1$.] There is $\kappa>0$ and $m = O(n)$ such that for $F \sim \XOR(m,n,k)$ the integrality gap of the $\kappa n$-round  $\sCP$ relaxation of $F$ is at least $(2-\varepsilon)$ with high probability. 
		\item[$2$.] There is $\kappa >0$ and $m=O(n)$ such that for $F \sim \mathcal{F}(m,n,k)$ the integrality gap of the $\kappa n$-round $\sCP$ relaxation of $F$ is at least $2^k/(2^k-1) - \varepsilon$ with high probability.
	\end{enumerate}
\end{thm}
\begin{proof}
	Let $F \sim \XOR(m,n,k)$ and let $Y_i$ be the event that the $i$th constraint is falsified by a uniformly random assignment. Let $\delta:=\varepsilon/(2-\varepsilon)$, then by a multiplicative Chernoff Bound, the probability that a uniformly random assignment satisfies at least a $1/(2-\varepsilon)$-fraction of $F$ is 
	$ \Pr [\sum_{i \in [m]}Y_i \geq (1+\delta)\frac{m}{2} ] \leq 2^{-\delta m/6}$.
	By a union bound, the probability that there exists an assignment satisfying at least a $1/(2-\varepsilon)$ fraction of $F$ is $2^{n-\delta m/6}$ which is exponentially small when $m \geq 7n(2-\varepsilon)/\varepsilon$.
	
	On the other hand, consider the partial restriction to the LP relaxation of $F$ that sets $y_i=1$ for all $i \in [m]$.
	Setting $m \geq 7n(2-\varepsilon)/\varepsilon$ large enough, by \autoref{thm:randomLB} there some $\kappa > 0$ such that with high probability $F$ requires depth $\kappa n$. 
	Hence, the $\kappa n$ round Semantic $\CP$ LP relaxation is non-empty, and there is a satisfying assignment $\alpha \in \mathbb{R}^n$. Thus $\alpha \cup \{y_i=1\}$ satisfies all constraints of $\max(F)$.
	
	The second result follows by an analogous argument.
\end{proof}



\section{Conclusion}
We end by discussing some problems left open by this paper. The most obvious of which is a resolution to \autoref{conj:deepProofs}. A related question is whether  supercritical size-depth tradeoffs can be established for monotone circuits? Indeed, current size lower bound techniques \cite{Pudlak97, FlemingPPR17,HrubesP17,GargGKS18} are via reduction to monotone circuit lower bounds. As a first step towards both of these, can one prove a supercritical size-depth tradeoff for a weaker proof system such as resolution?

The simulation results presented in \autoref{sec:SP} leave open several questions regarding the relationship between $\SP$ and $\CP$. First, the simulation of $\SP^*$ by $\CP$ incurs a significant blowup in the coefficient size due to Shrijver's lemma. It would be interesting to understand whether $\SP^*$ can be quasi-polynomially simulated by $\CP^*$; that is, whether this blowup in the size of the coefficients is necessary.

The most obvious question left open by these simulations is whether $\CP$ can polynomially simulate $\SP$, or even \emph{polynomially} simulate $\SP^*$. 
Similarly, what are the relationships of both $\SP$ and $\CP$, to (bounded-coefficient) $\RCP$, the system which corresponds to dag-like $\SP$. 
$\RCP$ can polynomially simulate DNF resolution, and therefore has polynomial size proofs of the Clique-Colouring formulas, for cliques of size $\Omega(\sqrt{n})$ and colourings of size $o(\log^2 n)$ \cite{AtseriasBE02}. 
Quasi-polynomial lower bounds on the size of $\CP$ refutations are known for this range of parameters and this rules out a polynomial simulation by Cutting Planes; however, a quasi-polynomial simulation may be possible. A potential approach to resolving this question is to use the added expressibility of $\RCP$ over DNF resolution to extend the upper bound on Clique-Colouring to the range of parameters for which superpolynomial $\CP$ lower bounds are known.

\subsection*{Acknowledgements}
N.F.~would like to thank Albert Atserias for some corrections to an earlier version of this paper.
T.P.~was supported by NSERC, NSF Grant No. CCF-1900460 and the IAS school of mathematics.
R.R.~was supported by NSERC, the Charles Simonyi Endowment, and indirectly supported by the National Science Foundation Grant No. CCF-1900460.  L.T.
was supported by NSF grant CCF-192179 and NSF CAREER award CCF-1942123.
Any opinions, findings and conclusions or recommendations expressed in this material are those of the author(s) and do not necessarily reflect the views of the National Science Foundation. 

\bibliographystyle{plain}
	\bibliography{biblio}

\end{document}